\documentclass[12pt]{article}
\usepackage{amsmath,amsfonts,amssymb,amsthm}
\usepackage{graphicx}
\usepackage[utf8]{inputenc}

\usepackage[authoryear]{natbib}

\newcommand{\R}{\mathbb{R}}
\newcommand{\N}{\mathbb{N}}
\newcommand{\strat}{\chi}
\newcommand{\rd}{\mathrm{d}}
\newcommand{\rv}{\mathrm{v}}
\newcommand{\opt}{{\mathrm{opt}}}

\newcommand{\radon}{\mathsf{RM}_+}
\newcommand{\nash}{{\mathrm{Nash}}}
\newcommand{\bydef}{:=}

\newcommand{\eps}{\varepsilon}
\newcommand{\e}{\mathrm{e}}

\newcommand{\PD}{\mathsf{P}}

\newtheorem{lemma}{Lemma}

\newtheorem{definition}{Definition}
\newtheorem{remark}{Remark}
\newtheorem{corollary}{Corollary}
\newtheorem{theorem}{Theorem}
\newtheorem{proposition}{Proposition}

\usepackage{tikz}
\usetikzlibrary{shapes.geometric, arrows,positioning}
\tikzstyle{estado} = [rectangle, rounded corners, minimum width=2cm, minimum height=2cm,font=\Huge,text centered, draw=black, fill=black!5]

\begin{document}

\title{Optimal Vaccination Strategies and Rational Behaviour in Seasonal Epidemics}
\author{Paula Rodrigues${}^*$, Paulo Doutor${}^*$, Maria do C\'eu Soares${}^*$, Fabio A. C. Chalub%
\thanks{Departamento de Matemática and Centro de Matemática e Aplicações, Faculdade de Ciências e Tecnologia, Universidade Nova de Lisboa, Quinta da Torre, 2829-516, Caparica, Portugal. e-mail:\{pcpr,pjd,mcs,chalub\}@fct.unl.pt}
}
\date{\today}

\maketitle

\begin{abstract}
We consider a SIRS model with time dependent transmission rate. We assume time dependent vaccination which confers the same immunity as natural infection. We study two types of vaccination strategies: i) optimal vaccination, in the sense that it minimizes the effort of vaccination in the set of vaccination strategies for which, for any sufficiently small perturbation of the disease free state, the number of infectious individuals is monotonically decreasing; ii) Nash-equilibria strategies where all individuals simultaneously minimize the joint risk of vaccination versus the risk of the disease. The former case corresponds to an optimal solution for mandatory vaccinations, while the second corresponds to the equilibrium to be expected if vaccination is fully voluntary. We are able to show the existence of both optimal and Nash strategies in a general setting. In general, these strategies will not be functions but Radon measures. For specific forms of the transmission rate, we provide explicit formulas for the optimal and the Nash vaccination strategies.

\end{abstract}



\section{Introduction}

Vaccination is the best response available in the control of most infectious diseases. A huge effort is put on the development of new and better vaccines. When humans are directly involved, the role of direct experimentation is naturally limited and therefore mathematical models have been used to evaluate the effect of control measures, such as vaccination, to assist in policy decisions. One central result of classical mathematical models for the spread of infectious diseases is that persistence of an infectious disease within a population requires the density of susceptible individuals to exceed a strictly positive critical value such that, on average, each primary case of infection generates more than one secondary case. It is therefore not necessary to vaccinate everyone within a community to eliminate infection. This phenomenon is known as herd immunity  and is one of the key epidemiological questions in defining a vaccination strategy.

In this work, we consider a SIRS model with periodic transmission. The model consists of a non-autonomous system of ordinary differential equations in which we introduce periodic vaccination of adults. For simplicity, we considered that vaccination confers the same protection as natural infection. We study the consequences of two extreme types of vaccination strategies: mandatory vaccination, where the population is vaccinated at a predefined rate; and voluntary vaccination, where individuals can choose freely to be vaccinated or not, according to their risk perception.

Mathematical models have  been widely used to help health authorities in the definition of vaccination strategies for very different contexts. Typically, the objective is to define an {\it optimal vaccination} strategy   by minimizing combinations of the vaccination effort/cost  and of the effective reproduction number  $\mathcal{R}_0$ (i.e., the number of secondary infections generated by a primary case) \citep{isham1996,castillo-chavez1998,LaguzetMB}. 
This can give rise to particularly interesting problems when we consider non-homogeneous models for which vaccination strategy depends on age \citep{isham1996,castillo-chavez1998,tartof2013}, risk-groups \citep{scott2015,long2011} or when it is time-dependent \citep{onyango2014,donofrio2002, browne2015,houy2016}. The current work is concerned with time-dependent epidemic models with vaccination when both the transmission rate and the vaccination are assumed to be periodic. Note that periodic vaccinations are used by public health services; e.g., influenza vaccine is only available in a specific season of the year.

Whenever the goal is long term disease elimination, optimal vaccination will consist on reducing $\mathcal{R}_0$ below one, in the sense that it implies the attractiveness and the asymptotic stability of the disease-free state. 
 Here, we choose to work with an alternative definition of optimal vaccination, where the goal is, not only to eliminate disease  but also to prevent outbreaks. Hence, we define a class of preventive vaccination  profiles such that, for any sufficiently small perturbation of the disease free state, the number of infectious individuals is monotonically decreasing. We construct the optimal vaccination strategy as the limit of preventive strategies for which  vaccination effort is minimized. We start by  revisiting basic concepts in mathematical epidemiology to recall that for constant transmission, the  condition $\mathcal{R}_0\le 1$, and subsequent stability of the disease-free state, is equivalent to the condition that  infectious population $I$ is monotonically decreasing in time, when the initial number of susceptible individuals is below the number of susceptible individuals in the disease free state. However, as we move towards more general situations this equivalence may not hold. Note that the former condition refers to the long term behaviour of the system, while the latter considers also the short time behaviour which, in principle, is more restrictive \citep{hastings2010,hastings2004}. In particular, the former condition restricts the average transmission over a period  and the latter is defined pointwise in time. Our approach is particularly suitable for diseases with high mortality or morbidity rates, for which it is imperative to prevent outbreaks. Despite that, so far our model and examples will not consider disease related death.

On the opposite end of vaccination policies is voluntary vaccination, which is increasingly common in industrialized countries. Even when vaccines are offered by the public health system without costs, vaccination is, at least in part, voluntary. Opposition to vaccines can be philosophical, religious and  depend on social contacts and information available. It puts important challenges to disease control by decreasing vaccine uptake. The case which is best known is the measles, since the unproven hypotheses that measles-mumps-rubella (MMR) vaccine was linked to autism led to a decrease in vaccination followed by measles epidemics in UK \citep{fitzpatrick2004,jansen2003}. Voluntary vaccination can also give rise to free-rider phenomenon, where individuals or families choose not to be vaccinated, or to not have their children vaccinated, taking advantage of herd immunity created in the population by others, avoiding the possible negative effects of vaccination. In this work, we consider a population of rational individuals that compares the risk of the vaccination (more precisely, its perception of the risk of the vaccination) and the risk of the disease and make options that minimizes the joint risk. Despite the fact that some countries are implementing fines for parents that prefer not to vaccinate their children\footnote{That's the case of Poland and Australia. See~http://www.thenews.pl/1/9/Artykul/204007, Parents-fined-for-not-vaccinating-children and http://naturalsociety.com/australia-enforces-15k-penalty-for-parents-who-dont-vaccinate/, respectively.}, we do not introduce in the model a risk of non-vaccination, other than the one associated to the disease.

In this work, we model human behaviour using game theory. In a seminal paper  by \citet{bauch2004}, it has been shown that voluntary vaccination cannot lead to disease eradication. The authors coupled a SIR model for disease spread in a partially vaccinated population with a theoretical game framework describing a population of rational individuals. Many subsequent developments were made in order to include the human behaviour in epidemiological models (cf. \citet{chen2006,dOnofrio2007,manfredi2009,coelho2009,
mbah2012,manfredi2013,morin2013,bhat2015,LaguzetBMB}). See also \citep{funk2010,wang2015} for a review, and \citep{funk2015} for further discussion on the subject.

In this paper, we generalize the framework of \citet{bauch2004} to the SIRS model with periodic transmission function. From the modelling point of view, we consider that all choices in the population influence the dynamics, and the resulting dynamics also has effect in the rational behaviour of the population. Due to the richness of the non-autonomous system that describes our model, several technical problems arise. For instance, the risk of disease no longer depends simply on the constant steady state as before. Considering a rational individual, we assume that he/she is going to choose to be vaccinated only when the risk of disease times the probability of being infected is higher than the risk of the vaccine, as perceived by the taker. As we analyse only stationary states of our periodic system, the risk to be minimized is the joint risk of vaccination and disease during one season. Hence, we define the set of herd immunity provider vaccination strategies, for which the rational strategy for a given focal (rational) individual is not to be vaccinated, taking advantage of the herd immunity provided by the choices of the rest of the population. Moreover, we define a Nash vaccination strategy as the strategy that minimizes the joint risk for every individual taking into account the strategy of all other individuals, i.e., the natural strategy to be expected in a population of rational individuals with full knowledge of all epidemiological data.

Existence of optimal and Nash vaccination strategies are proved in this work in a very general setting; however,  these strategies may not be functions but Radon measures, even for  transmission rates given by a real function. This is a consequence of the fact that the set of continuous functions in a given compact interval is not closed under any reasonable metric. Many results used in the existence proofs presented in the appendices require compactness and after introducing a convenient topology in the set of continuous functions, we are naturally led to the introduction, in this framework, of Radon measures. For more information on the topic of Radon measures we refer to~\citep{schwartz1973,athreya2006}.

The paper is organized as follows. In Section \ref{sec:model}, we introduce the mathematical model and derive some preliminary results. Section~\ref{sec:vaccinations} is dedicated to the vaccination strategies. First, we give rigorous definitions of preventive vaccination strategies, and of  vaccination effort and we define the optimal strategy as one strategy that can be arbitrarily approximated by a preventive strategy and such that its associated effort is never superior to the effort of any given preventive strategy.  In the context of voluntary vaccination, we define the set of herd immunity provider strategies and the concept of Nash strategy, in which all individuals minimizes the joint risk of vaccination and disease. In the end of the Subsection~\ref{ssec:rational}, we state the main theorem, which guarantees the existence of an optimal and a Nash vaccination strategies in the set of Radon measures. Explicit formulas for the optimal and Nash strategies are provided in Subsection~\ref{sec:regular}, for specific forms of the transmission rate.  In Section \ref{sec:examples}, we present some examples such as the constant transmission case, the sinusoidal case and also a critical case to illustrate the results from previous sections. {We finish with two appendices, the first one guaranteeing the existence of periodic solutions in the model and the second proving the existence of optimal and Nash strategies.

\section{The Model}\label{sec:model}

Consider a SIRS model. Let $S(t)$, $I(t)$, $R(t)$ be the fraction of susceptible, infectious and recovered individuals at time $t\ge 0$. We assume non negative normalized initial conditions, i.e, $S(0),I(0),R(0)\ge 0$, $S(0)+I(0)+R(0)=1$. We also assume the transitions $S+I\stackrel{\beta}{\longrightarrow}2I$, $I\stackrel{\gamma}{\longrightarrow}R$, $R\stackrel{\alpha}{\longrightarrow}S$, $S\stackrel{p}{\longrightarrow}R$. 
Constants $\mu$ (mortality/birth rate), $\alpha$ (temporary immunity) and $\gamma$ (recovery rate) are strictly positive. 
The disease is assumed to be non-fatal, i.e., the death rate $\mu>0$ does not depend on the disease class. By normalization, we also consider the birth rate as $\mu$. These are common assumptions of compartmental epidemiological models. 

We consider functions $\beta,p:\R_+\to\R_+$, representing the transmission and vaccination rates at time $t$, respectively. More precise assumptions on these functions will be introduced latter on.

From now on, we call SIRS model to the following system of differential equations:
\begin{subequations}\label{system}
\begin{align}\label{systemS}		
			S'&=\mu+\alpha R-\beta(t)IS-p(t)S-\mu S \nonumber \\&=\mu+\alpha-\alpha I-\beta(t)IS-p(t)S-(\mu+\alpha)S \\
\label{systemI}
			I'&=\beta(t)IS-\gamma I - \mu I \\
\label{systemR}
			R'&=\gamma I + p(t) S-\mu R-\alpha R 
\end{align}
\end{subequations}
 A schematic representation of the SIRS model with vaccination is represented in Figure~\ref{fig:SIR}. Due to the normalization $S(t)+I(t)+R(t)=S(0)+I(0)+R(0)=1$, the equation for $R$ is always redundant and will be ignored from now on. We define 
 $\Delta^2\bydef\{(x,y)\in\R^2|\,x\ge 0,y\ge 0,x+y\le 1\}$.\\

\begin{figure}
\centering 
\begin{tikzpicture}
\node (ss) [estado] {S};
\node (ii) [estado,right of=ss,xshift=3cm] {I};
\node (rr) [estado,right of=ii,xshift=3cm] {R};

\draw[line width=1.5pt,->] (ss) -- node[above=0.1cm,font=\Large] {$\beta(t) I$} (ii);

\draw[line width=1.5pt,->] (ii) -- node[above=0.1cm,font=\Large] {$\gamma$} (rr);

\draw[line width=1.5pt,->] (-2cm,0cm) -- node[above=0.1cm,font=\Large] {$\mu$} (ss);
\draw[line width=1.5pt,->] (ss) -- node[right=0.1cm,yshift=-0.5cm,font=\Large] {$\mu$} (0cm,-2.5cm);
\draw[line width=1.5pt,->] (ii) -- node[right=0.1cm,yshift=-0.5cm,font=\Large] {$\mu$} (4cm,-2.5cm);
\draw[line width=1.5pt,->] (rr) -- node[right=0.1cm,yshift=-0.5cm,font=\Large] {$\mu$} (8cm,-2.5cm);

\draw[line width=1.5pt,->] (0cm,1cm) arc (150:30:4.5cm and 1.5cm);

\draw(1.5cm,1.8cm) node[font=\Large] {$p(t)$};

\draw[line width=1.5pt,->] (8cm,-1cm) arc (-30:-150:4.5cm and 1.5cm);

\draw(7cm,-1.8cm) node[font=\Large] {$\alpha$};

\end{tikzpicture}
\caption{Schematic diagram of the SIRS model with death and vaccinations.}
\label{fig:SIR}
\end{figure}
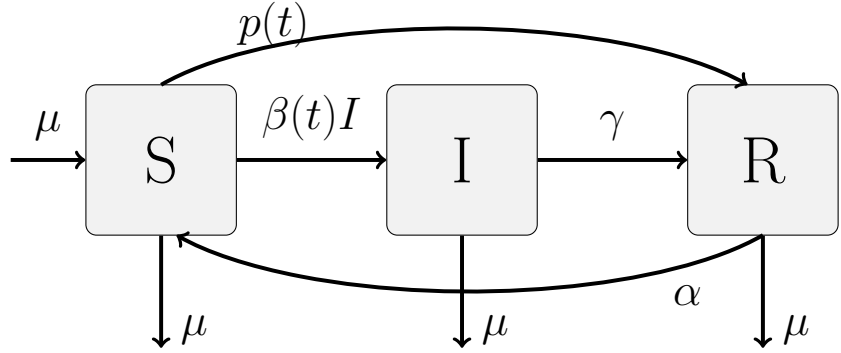

We begin by analysing the solutions and stability of system \eqref{system}.

\begin{lemma}\label{lem:sol_per}
Let us consider that functions $\beta$ and $p$ are continuous functions with commensurable periods, i.e., there exists $T>0$ such that $p(t+T)=p(t)$ and $\beta(t+T)=\beta(t)$ for all $t$. 
Equivalently, we assume that $\beta,p \in C([0,T])$ with $(p(0),\beta(0))=(p(T),\beta(T))$.

Therefore, there exists only
one periodic solution of system~\eqref{system} in the subspace $\{I(t)=0,\ \forall t\}$, given by $(S_0(t),0)\in\Delta^2$. We call this solution the \emph{disease-free} solution. 
This solution attracts all initial conditions of the form $(S(0),0)$. We define $I_0(t)=0$.

Depending on the choices of the parameters $\gamma,\mu$ and the functions $p,\beta$, we may have one of two possibilities:
\begin{enumerate}
 \item The disease-free solution is globally stable in $\Delta^2$;
 \item There are other periodic solutions (with period multiple of $T$), called \emph{endemic} solutions $(S_i(t),I_i(t))\in\Delta^2$, with $I_i(t)>0$, for all $t\in\R_+$ and  $i\in\mathbb{N}$. In this case, there is $\eta>0$ such that for any 
initial condition $(S(0),I(0))$ with $I(0)>0$, we have $\liminf_{t\to\infty} I(t)> \eta$. In this case, we say that the solution of the SIRS model is \emph{persistent}.
\end{enumerate}
Furthermore, for any initial condition, the solution $(S(t),I(t))$ depends continuously on $\beta$ and $p$; namely, if $p_n\to p$ and $\beta_n\to\beta$ weakly as measures and both sequences are uniformly integrable, then $S_i[p_n,\beta_m]\to S_i[p,\beta]$ and $I_i[p_n,\beta_m]\to I_i[p,\beta]$ uniformly in $[0,T]$, when $n,m\to\infty$, $i\in\mathbb{N}$.
\end{lemma}

\begin{proof}
For the first part, see~\citet{rebelo2012}; see Appendix~\ref{ap:proof1} for further details. The second part follows from~\citet[theorem 2.1]{heunis1984}. 
\end{proof}

\begin{remark}
 The assumptions on $\beta,p $ in the previous lemma are extremely restrictive and used only for the first part of the result (existence of disease free solution and periodicity of the endemic solution). 
 If we relax our assumptions to require only that $\beta$ is of bounded variation and $p$ is a measurable function, then existence of solutions (not necessarily periodic) and convergence of solutions (as in the second part of Lemma~\ref{lem:sol_per}) is guaranteed by~\citep{heunis1984}. This will be explored in the examples. Note that if $\beta_n,p_n$ are continuous then, $\beta\bydef\lim\beta_n$ and $p\bydef\lim p_n$ necessarily satisfy these more relaxed assumptions. 
\end{remark}

For constant transmission and vaccination, we  establish the following result, that is going to motivate our definition of optimal vaccination.

\begin{lemma}\label{lem:equiv}
  Consider that $\beta(t)=\beta_0>0$ and $p(t)=p_0\ge 0$.  The only stationary disease free solution of  system~\eqref{system} is given by $\hat S_0=\hat S_0[p]=S_0[p](t)=\frac{\mu+\alpha}{p_0+\mu+\alpha}\le 1$. Furthermore, the three conditions below are equivalent:
 \begin{enumerate}
  \item\label{lem:equiv_a} $\frac{\beta_0\hat S_0}{\gamma+\mu} \leq 1$.
  \item\label{lem:equiv_b} The disease free solution $(\hat S_0,0)$ is globally asymptotically stable.
  \item\label{lem:equiv_c} $I'<0$ for all $I>0$ and all $S<\hat S_0$.
 \end{enumerate}
  \end{lemma}

  \begin{proof}
  	
  	\ref{lem:equiv_a}$\Leftrightarrow$\ref{lem:equiv_c}. Note that $I'=I(\gamma+\mu)\left(\frac{\beta_0 S}{\gamma+\mu}-1\right)$, and therefore, assuming $I>0$, $I'<0$ if and only if $S<\frac{\gamma+\mu}{\beta_0}$. 
  	
  	\qquad\ref{lem:equiv_a}$\Rightarrow$\ref{lem:equiv_c}. $S<\hat S_0\le\frac{\gamma+\mu}{\beta_0}$ and then $I'<0$. 
  	
  	\qquad\ref{lem:equiv_c}$\Rightarrow$\ref{lem:equiv_a}. $I'<0$ if and only if $S<\frac{\gamma+\mu}{\beta_0}$ and therefore $\hat S_0\leq\frac{\gamma+\mu}{\beta_0}$.
  	
	\ref{lem:equiv_a}$\Leftrightarrow$\ref{lem:equiv_b}.
	
  	\qquad\ref{lem:equiv_a}$\Rightarrow$\ref{lem:equiv_b}.
We follow ideas from~\citep{Cruz}; see also~\citep{Capasso} for other examples of use of Lyapunov functions in Mathematical Epidemiology. 

Let us define 
  	\[
 	V(S,I)=\frac{1}{2}\left(S-\hat S_0+I\right)^2+ a I, \quad a=\frac{2(\mu+\alpha)+\gamma +p_0}{\beta_0}.
  	\]
	We differentiate $V$ with respect to $t$ and obtain
  	\[
  	\begin{split}
  	V'(S,I)=\frac{dV}{dt}=&\left(S-\hat S_0+I\right)(S'+I')+a I'\\
  	=&\left(S-\hat S_0+I\right)( (\hat S_0-S)(p_0+\mu+\alpha) -(\gamma +\alpha+ \mu) I)+a I'\\
  	=&-\left(S-\hat S_0\right)^2(p_0+\mu+\alpha)-I^2(\gamma +\alpha+ \mu)\\&- (S-\hat S_0)I(2(\mu+\alpha)+\gamma +p_0) +a( \beta_0 IS-(\gamma+ \mu)) I.\\
	=&-\left(S-\hat S_0\right)^2(p_0+\mu+\alpha)-I^2(\gamma +\alpha+ \mu)\\&-a \beta_0 I \left(\frac{\gamma+ \mu}{\beta_0}-\hat S_0\right).\\
 	\end{split}
 	\]
  	Let $G=\{(S,I)\in[0,1]\times (0,1]|S+I\le 1\}$. Note that $V'$ is a continuous function in $G$ and, by Condition~\ref{lem:equiv_a}, it is negative. It is immediate to verify that $V$ is a Lyapunov function in $G$. Let $\bar{G}$ be the closure of $G$. As  $\{(S,I) \in \bar{G}| V'(S,I)=0\}$ is the singleton with the equilibrium point $\{(\hat S_0,0)\}$ we conclude from \citet[Corollary~1.2 in Chapter~X]{Hale} that $(\hat S_0,0)$ is globally asymptotically stable.

 	\qquad\ref{lem:equiv_b}$\Rightarrow$\ref{lem:equiv_a}. After system linearisation around the disease free solution $(\hat S_0,0)$, we find the Jacobian matrix 
  	\[
  	\left(\begin{matrix} -p_0-\mu-\alpha&-\alpha-\beta_0\hat S_0\\0&(\gamma+\mu)\left(\frac{\beta_0\hat S_0}{\gamma+\mu}-1\right)\end{matrix}\right)\ .
  	\]
  	If $\frac{\beta_0\hat S_0}{\gamma+\mu} > 1$, the equilibrium would not be stable, which leads to a contradiction.

  \end{proof}

From the above lemma, we recover the effective reproductive number for the constant parameter case, $\mathcal{R}_0\bydef\mathcal{R}_{0}[p]\bydef\frac{\beta_0\hat S_0[p]}{\gamma+\mu}$. Condition $\mathcal{R}_{0}\leq 1$ guarantees at the same time that all epidemics will be eventually extinct and that $I(t)$ decreases monotonically in time, from $I(0)>0$.

However, in the time dependent case (in particular in the periodic case), these two phenomena are not equivalent. In general, even for linear systems, it is possible that before being attracted to an asymptotic equilibrium, the trajectory of $(S(t),I(t))$ drifts away from this equilibrium~\citep{hastings2010,hastings2004}.

For the periodic case, we can compute the effective reproduction number, following~\citep{Thieme_2000} (see also~\citep{wang2008}), as 
\begin{equation}\label{def:R0}
\mathcal{R}_0\bydef\frac{1}{\gamma+\mu}\langle\beta S_0\rangle=\frac{1}{T(\gamma+\mu)}\int_0^T\beta(t)S_0(t)\rd t\ .
\end{equation}
Note that, for the periodic case,  condition  $\mathcal{R}_0<1$  still guarantees  asymptotic stability of the disease free case \citep{wang2008}, but does not necessarily  prevent the existence of outbreaks; see, e.g., \citep{zhao2008}.

In this work, we will look for conditions that generalize, for time-dependent parameters, Condition~\ref{lem:equiv_c} in Lemma~\ref{lem:equiv}, i.e., that guarantees that the number of infectious is monotonically decreasing for small perturbations of the disease free solution. From the modelling point of view, no particular definition can be considered better than the other; in fact, for certain particular diseases (e.g., polio, tuberculosis) vaccination policy aims to eradicate/eliminate the disease in the long run, while for other diseases, governments act to prevent the existence of large outbreaks (e.g., influenza, cholera)~\citep{WHO2015}. Our approach describes better this second setting.\\

 From now on, we assume that, for a given vaccination strategy $p(t)$,  system $(S(t),I(t))$ is in its stationary (periodic) state, and we will consider two different cases:
 \begin{enumerate}
 \renewcommand{\theenumi}{\arabic{enumi}}
\renewcommand{\labelenumi}{(C\theenumi)}
 \item
 \label{case1} The disease free state $(S_0[p](t),0)$;
  \item
\label{case2} A certain endemic state $(S_1[p](t),I_1[p](t))$. (There is no uniqueness for the endemic state; for the sake of simplicity, we will consider from now on only one endemic solution. There is no essential change if we consider more than one.)
 \end{enumerate}
Both solutions are assumed to be periodic, possibly with period multiple of $T$; however, without loss of generality, we will consider the period given by $T$. Note that for a different set of parameters more complicated behaviour (possible chaotic) can be found, cf. \citep{Kuznetsov_Piccardi}.

\section{Vaccination strategies}\label{sec:vaccinations}

In this section we will consider two types of vaccination: mandatory and voluntary vaccination. For each one, we will define one special case: for the former, an optimal vaccination is defined as one vaccination strategy that is  able to prevent outbreaks while having the minimum number of vaccinations possible and for the latter, a Nash vaccination strategy is defined as a strategy in which all individuals in population minimize the joint risk of both disease and vaccine.

\subsection{Optimal vaccination}

For the optimal vaccination, we choose to work with  a generalization of Condition~\ref{lem:equiv_c} in Lemma~\ref{lem:equiv}. More specifically, we say that a certain vaccination strategy $p$ is a {\it preventive strategy}  when the fraction of individuals in the class $I$ decreases monotonically in time for any  small enough perturbation of the disease free state. We then construct the optimal vaccination strategy as the limit of the preventive strategies for which the vaccination effort is minimized. In \citep{onyango2014}, optimality for time dependent vaccination profiles is defined based on the effective reproductive number. Note that in our model, only susceptibles are vaccinated, which implies a full knowledge of the current status of an individual.

\begin{definition}\label{def:eff}
We define the \emph{vaccination effort} associated to a given strategy $p$, as the average number of vaccinations in one period, i.e. $\mathbb{E}[p]\bydef\langle p S[p]\rangle\bydef \frac{1}{T}\int_0^T p(t)S[p](t)\rd t$, where $\langle\cdot\rangle$ denotes the average in one period, and $S$ is the relevant solution, given by (C\ref{case1}) or (C\ref{case2}), defined above. 
\end{definition}

We denote a cumulative distribution function, associated with $p$, by $\PD(t)=\int_0^t\rd p$, or in a more relaxed notation $\rd\PD(t)=p(t)\rd t$. To simplify the notation, we will use indistinctly $\rd\PD$ and $p\rd t$, whenever there is no risk of confusion. Therefore, we now write $\mathbb{E}[p]=\frac{1}{T}\int_0^TS(\tau)\rd\PD(\tau)$.
For technical reasons, we need to consider bounds in the set of vaccination profiles. 
More precisely:

\begin{definition}\label{deff:admissible}
We say that a certain vaccination function  $p$ is \emph{admissible} if its cumulative distribution is such that
\begin{equation}\label{eq:hip_p}
\PD([0,T])=\int_{0}^T\rd \PD\le (\mu+\alpha\ )T\frac{\bar\beta}{\gamma+\mu},
\end{equation}
where $\bar\beta=\sup_{t\in[0,T]}\beta(t)$.
Furthermore, we use $\radon$ to denote the set of non-negative Radon measures $\PD$ in $[0,T]$ such that $\PD([0,T])\le (\mu+\alpha\ )T\frac{\bar\beta}{\gamma+\mu}$ and $C_+([0,T])$ to denote the set of continuous functions in $[0,T]$, with $p(0)=p(T)$ such that $\int_0^Tp(t)\rd t\le (\mu+\alpha)T\frac{\bar\beta}{\gamma+\alpha}$. We also consider the natural immersion $C_+([0,T])\subset\radon$. 
 \end{definition}

Now we show that the definition of vaccination effort can be extend continuously for the case of Radon measures.

\begin{lemma}\label{lem:effort_measure}
 Let $p_n\in C_+([0,T])$ be such that $p_n\to p\in\radon$, in the weak topology, cf.~\citep{koralov2007}, and let $\mathbb{E}[p]\bydef\lim\mathbb{E}[p_n]$. Then, $\mathbb{E}[p]$ is independent of the choice of the sequence $p_n$.
\end{lemma}

\begin{proof}
Let $p_n,q_n\in C_+([0,T])$ such that $p_n,q_n\to p\in\radon$. Let $\PD_n$ and $\mathsf{Q}_n$ be the associated cumulative distribution functions, respectively. Note that  
 \begin{align*}
  \mathbb{E}[p_n]-\mathbb{E}[q_n]&=\int_0^T \left(S[p_n]\rd \PD_n-S[q_n]\rd \mathsf{Q}_n\right)\\
  &=\int_0^T S[p_n]\rd\left(\mathsf{P}_n-\mathsf{Q}_n\right)+\int_0^T\left(S[p_n]-S[q_n]\right)\rd \mathsf{Q}_n\ .
 \end{align*}
 From the fact that $S[p_n]$ is bounded and $\mathsf{P}_n-\mathsf{Q}_n\to 0$, we conclude that the first integral converges to 0. 
 For the second integral, the convergence to 0 follows from the continuity of $p\mapsto S[p]$ in the appropriate topology. See~\citep[theorem 2.1]{heunis1984} for further details.
\end{proof}

\begin{definition}\label{deff:preventive}
Let $\beta\in C([0,T])$, $\beta(0)=\beta(T)$ be given.
 For a given vaccination strategy $p$, let $(S_0(t),0)=(S_0[p](t),0)$ be the disease-free 
 solution of System~(\ref{system}).  
We say that  $p$ is a \emph{preventive strategy} if $\beta(t)S_0[p](t)< \gamma+\mu$ for all $t$. We call $\strat_{\mathrm{p}}=\strat_{\mathrm{p}}[\beta]$ the set of admissible strategies that are preventive, i.e,
$
\strat_{\mathrm{p}}=\{p \in C_+([0,T]),\ \text{with}\ p(0)=p(T),\ \text{and}\  \beta(t)S_0[p](t)< \gamma+\mu,$ $ \textrm{ for all } t \in [0,T]\}.
$
 
\end{definition}

Now, we analyse the preventive strategies. First, we explicitly characterize the disease free state and then we show the existence of at least one preventive strategy. Afterwards, we define the concept of optimal strategy. Here, we reproduce the result from~\citep[Theorem 3.7]{Thieme_2003}.

 \begin{lemma}\label{lem:S00}
Let $S_0[p](t)$ be the time dependent periodic number of susceptibles in the unique disease free state of system~\eqref{system}. Then
  \[
   S_0[p](0)=\frac{(\mu+\alpha)\int_0^T\e^{-\int_s^T(p+\mu+\alpha)(\tau)\rd\tau}\rd s}{1-\e^{-\int_0^T(p+\mu+\alpha)(\tau)\rd\tau}}\ .
  \]
 \end{lemma}
 
Before looking for optimal strategies, we prove that the set of preventive strategies is not empty.

\begin{lemma}\label{lem:nonempty}
For any choice of parameters, there exists at least one preventive strategy, i.e.,
 $\strat_{\mathrm{p}}[\beta]\ne \emptyset$ for all $\beta\in C([0,T])$, $\beta(0)=\beta(T)$.
\end{lemma}

\begin{proof}
Given $\beta$, assume $p(t)=p_0>(\mu+\alpha)\left(\frac{\bar\beta}{\gamma+\mu}-1\right)$, constant, corresponding to a preventive strategy in the case of the maximum transmission rate.
Note that $S_0[p](t)=\frac{\mu+\alpha}{p_0+\mu+\alpha}$ is the only stationary solution of the Equation~(\ref{systemS}) with $I(t)=0$ and is, additionally, the solution obtained from the initial condition given by Lemma~\ref{lem:S00}. From the definition of $p_0$, we conclude that $\beta(t)S_0[p](t)<\gamma+\mu$ for all $t$, and then $\{p(t)=p_0\}\in\strat_{\mathrm{p}}$. 
\end{proof}

Finally, we construct the {\it optimal vaccination strategy}  as one strategy that can be arbitrarily approximated by a preventive strategy and such that its associated effort is never superior to the effort of any preventive strategy.
More rigorously, we define an optimal vaccination strategy by
\begin{definition}\label{def:popt}
 Let $\beta\in C([0,T])$, $\beta(0)=\beta(T)$ be given. We say that a given strategy $p_\opt=p_\opt[\beta]$ is optimal if the following conditions are simultaneously satisfied:
 \begin{enumerate}
  \item there is at least one sequence $\strat_{\mathrm{p}}[\beta]\ni p_n\to p_\opt$ (in measure).
 \item for any $p\in\strat_{\mathrm{p}}[\beta]$, $\mathbb{E}[p]\ge\mathbb{E}[p_\opt]$.
 \end{enumerate}
 \end{definition}

 Whenever $p_\opt$ is a function (as discussed in Subsection~\ref{sec:regular}), Definition~\ref{def:popt} means simply that there is a sequence of preventive strategies that converge to $p_\opt$ and that $p_\opt$ minimizes the vaccination effort in the closure of the set $\strat_{\mathrm{p}}$.

\subsection{Rational vaccination}\label{ssec:rational}

In this subsection, we study a population of rational individuals and how their decisions influence the disease dynamics.

For each focal individual the probability of getting the disease is assumed to depend on the disease incidence for each time step. The following lemma shows how this probability can be computed from the model.

\begin{lemma}
The probability that a susceptible non-vaccinated individual at time $t$ gets the disease between times $t$ and $t+\Delta t$, for $\Delta t$ sufficiently small, is given by $\beta(t)I(t)\Delta t+\mathrm{o}\left(\Delta t\right)$.
\end{lemma}

\begin{proof}
All susceptible non-vaccinated individuals are in the category $\textbf{S}$.
From time $t$ to time $t+\Delta t$, $\beta(t)I(t)S(t)\Delta t$ individuals will be infected, $\mu S(t)\Delta t$ will die and the remainder $S(t)-\left[\mu S(t)+\beta(t)I(t)S(t)\right]\Delta t$ will be in the class \textbf{S} at time $t+\Delta t$. Therefore, the probability to be infected from times $t$ to $t+\Delta t$, given that he/she did not die, is given by
\[
 \frac{\beta(t)I(t)S(t)\Delta t}{S(t)-\left[\mu S(t)+\beta(t)I(t)S(t)\right]\Delta t+\beta(t)I(t)S(t)\Delta t}=\beta(t)I(t)\Delta t+\mathrm{o}\left(\Delta t\right),
\]
where we used $(1-\mu \Delta(t))^{-1}=1+\mathrm{o}\left(\Delta t\right)$.

See Figure~\ref{fig:racional}, for a schematic representation of this reasoning.
\end{proof}

 \begin{figure}
 \centering
\begin{tikzpicture}
 
\node (st) [estado] {$S(t)$}; 
\node (st+) [estado,xshift=9cm,yshift=1.5cm] {$S(t+\Delta t)$};
\node (it+) [estado,xshift=9cm,yshift=-1.5cm] {$I(t+\Delta t)$};

\draw[line width=1.5pt,->] (st) -- node[above=0.1cm,sloped,below,font=\large] {$\beta(t)I(t)S(t)\Delta t$} (it+);
\draw[line width=1.5pt,->] (st) -- node[right=0.1cm,sloped,above,font=\large] {$S(t)-\left[\mu+\beta(t)I(t)\right]S(t)\Delta t$} (st+);
\draw[line width=1.5pt,->] (st) -- node[right=0.05cm,font=\Large] {$\mu S(t)\Delta t$} (0cm,-2.5cm);
 
\end{tikzpicture}
 \caption{Transitions of non-vaccinated individuals from state \textbf{S} at time $t$ out by \emph{death} (down arrow) and to states \textbf{S} and \textbf{I} at time $t+\Delta t$. Note that indications in the arrows are for the total number of individuals leaving state \textbf{S} during interval $\Delta t$.}
 \label{fig:racional}
 \end{figure}
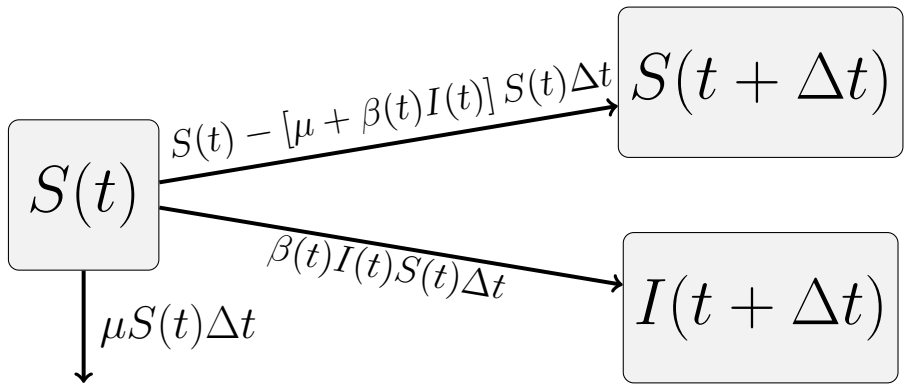

For the voluntary vaccination, we consider that a rational individual will (not) vaccinate him/herself if the risk of the disease times the probability to get the disease, given the overall strategy of the population, is larger than (respectively, small than) the risk of the vaccine. If both risks are the same, any strategy is equally advantageous. A fully informed rational individual will access, in the beginning of the season, the probability to get the disease, using all available epidemiological data, and decides his/her personal strategy as the strategy that minimizes the joint risk, i.e., the risk of the disease times the probability to get it (conditional to no vaccination), plus the risk of the vaccine (conditional to vaccination), during the next season. 
 
We start by defining the set of immunity provider strategies, i.e., the set of strategies for which a focal rational  individual will decide to be not vaccinated.

\begin{definition}\label{deff:herd_strategies}
Let $\beta\in C([0,T])$, $\beta(0)=\beta(T)$ be given.
 For a given vaccination strategy $p$, assume the existence of a persistent endemic solution $(S_1,I_1)$. Let $r_{\rd}>0$ and $r_\rv>0$ be the  risks of the disease and of the vaccination, respectively. We define $r\bydef\frac{r_\rv}{r_\rd}$. We say that
  $p$ is a  \emph{herd immunity provider strategy} if $\beta(t)I_1[p](t)< r$ for all $t$. We call $\strat_{\mathrm{h}}=\strat_\mathrm{h}[\beta]$ the set of all herd immunity provider strategies,
i.e.
\[
\strat_{\mathrm{h}}\!=\!\left\{p \in C_+([0,T])\ \text{with}\ p(0)\!=\!p(T),\ \text{and}\ \beta(t)I_1[p](t)<r, \textrm{for all } t \in [0,T]\right\}.
\]
If there is no endemic solution, we define $\strat_{\mathrm{h}}\!=\!\left\{p \!\in\! C_+([0,T]) \!\text{ with}\ p(0)\!=\!p(T)\right\}$.
\end{definition}

Note that, from the definition, it is clear that any preventive strategy is also herd immunity provider, i.e., $\strat_{\mathrm{p}}\subset\strat_{\mathrm{h}}$.

Finally, we will define the Nash-equilibrium strategy as the strategy that minimizes the joint risk for every individual, given the strategy of all other individuals.\\

\begin{definition}\label{def:nash_equilibrium}
Let $\beta\in C([0,T])$, $\beta(0)=\beta(T)$ be given.
  Let us consider a population with strategy $p\in C_+([0,T])$, and a focal individual that uses vaccination strategy $p_*\in\radon$. Let $\PD$ and $\PD_*$ be the cumulative distributions, associated to $p$ and $p_*$, respectively. Assume that the focal individual is susceptible at time $t=0$, and therefore the probability to be susceptible at a later time $t$ is given by $\e^{-\int_0^t\rd P_*-\mu t}$. The joint (disease and vaccination) risk during one season (i.e, the probability that something \emph{bad} --- disease or reaction to the vaccine --- happens in one season, times the associated risks) is given by 
 \begin{align*}
 \rho[p_*,p]
 &=r_\rd\int_0^T\beta(t)I[p](t)\e^{-\int_0^t\rd\PD_*-\mu t}\rd t+r_\rv\int_0^T\left(1-\e^{-\int_0^t\rd\PD_*-\mu t}\right)\rd t\\
 &=-r_\rd\int_0^T\left(r-\beta(t)I[p](t)\right)\e^{-\int_0^t\rd\PD_*-\mu t}\rd t+r_\rv T\ .
\end{align*}
 Given a strategy $p$, a rational individual will choose a strategy $p_*$ such that for every strategy $p'\in\radon$
 \[
  \rho[p_*,p]\le\rho[p',p]\ .
 \]
We say that $p_\nash\in\radon$ is a Nash strategy if for any sequence $p_n\in C_+([0,T])$, such that $p_n\to p_\nash$ and for every strategy $p'\in\radon$, 
\[
\limsup\left(\rho[p_\nash,p_n]-\rho[p',p_n]\right)\le 0 .
 \]
\end{definition}

If $p_\nash$ is a function, the above definition simplifies to the assertion that $\rho[p_\nash,p_\nash]\le\rho[p',p_\nash]$ for every strategy $p'\in\radon$.

We finish the subsection stating the existence theorem for both optimal and Nash-equilibrium strategies. In general terms, for $\beta\in C([0,T])$, with $\beta(0)=\beta(T)$, we prove that there is at least one optimal vaccination strategy and at least one Nash vaccination strategy. These strategies may not be functions, but measures. This implies that, after rewriting System~(\ref{system}) in the form $X'=\Gamma(t,X)$, the function $\Gamma:\R_+\times\Delta^2\to\R^2$ is a Charatheodory function (i.e., measurable in the first variable and continuous in the second) and therefore there is a (weakest) topology which guarantees existence of solutions of the differential equations and gives continuous dependence for each initial data point. See~\citep{Charalambos,heunis1984} for further details. The proof of the existence theorem below will be postponed to Appendix~\ref{ap:existence}.

\begin{theorem}\label{thm:existence}
 Assume $\beta\in C([0,T])$, with $\beta(0)=\beta(T)$. Then, there is at least one optimal vaccination strategy $p_\opt[\beta]$ and at least one Nash vaccination strategy $p_\nash[\beta]$.
\end{theorem}

\subsection{Vaccination strategies for regular transmission functions}\label{sec:regular}

 Despite the fact that we cannot guarantee \emph{a priori} existence of optimal and Nash strategies as functions, we will provide precise conditions for which $p_\opt$ and/or $p_\nash$ are functions. In particular, we derive explicit formulas for the optimal and Nash strategies for sufficiently regular transmission functions $\beta$, with some extra technical conditions. In the end, we discuss vaccination strategies when $\beta$ is discontinuous (in particular of bounded variation).

We start by finding an explicit formula for $p_\opt$ in some special cases.
 
\begin{theorem}\label{thm:opt_solution}
 Consider system~\eqref{system}. Assume that  
 \begin{equation}\label{eq:optcond}
\beta'(t)\ge-(\mu+\alpha)\beta(t)\left(\frac{\beta(t)}{\gamma+\mu}-1\right).
 \end{equation}
Then
 \begin{equation}\label{eq:popt}
p_\opt(t)=(\mu+\alpha)\left(\frac{\beta(t)}{\gamma+\mu}-1\right)+\frac{\beta'(t)}{\beta(t)}\ 
\end{equation}
is an optimal strategy.
\end{theorem}

\begin{proof}

First note that $p_\opt(t)\ge 0$ if and only if $\beta(t)$ satisfies Equation~\eqref{eq:optcond} for all $t\ge 0$.

We divide the proof in several steps:

$1^{\textrm{st}}$ step: We start by using Lemma~\ref{lem:S00} to show that $\beta(0)S[p_\opt](0)=\gamma+\mu$. Indeed, let $p=p_\opt$ and therefore
\begin{align*}
 S_0[p_\opt](0)&=\frac{(\mu+\alpha)\int_0^T\e^{-\frac{\mu+\alpha}{\gamma+\mu}\int_s^T\beta(\tau)\rd\tau-\log\frac{\beta(T)}{\beta(s)}}\rd s}{1-\e^{-\frac{\mu+\alpha}{\gamma+\mu}\int_0^T\beta(\tau)\rd\tau}}\\
 &=\frac{\gamma+\mu}{\beta(T)}\times\frac{\int_0^T\frac{\rd\ }{\rd s}\e^{-\frac{\mu+\alpha}{\gamma+\mu}\int_s^T\beta(\tau)\rd\tau}\rd s}{1-\e^{-\frac{\mu+\alpha}{\gamma+\mu}\int_0^T\beta(\tau)\rd\tau}}
 =\frac{\gamma+\mu}{\beta(0)}\ .
\end{align*}

$2^{\textrm{nd}}$ step: Now, we show that for any $t>0$, $\beta(t)S_0[p_\opt](t)=\gamma+\mu$.  
Using Equation~(\ref{systemS}) with $S(t)=S[p](t)$, $I(t)=I[p](0)=0$ and $p=p_\opt$, we find
\begin{align*}
&\left(\beta(t)S(t)\right)'\\
&\quad =\beta'(t)S(t)+\beta(t)S'(t)\\
&\quad=\beta'(t)S(t)+\beta(t)(\mu+\alpha)-\left[\beta(t)(\mu+\alpha)\left(\frac{\beta(t)}{\gamma+\mu}-1\right)+\beta'(t)\right]S(t)\\
&\qquad-\beta(t)(\mu+\alpha) S(t)\\
&\quad=\beta(t)(\mu+\alpha)\left(1-\frac{\beta(t)S(t)}{\gamma+\mu}\right).
\end{align*}
We conclude that  $\beta(t)S(t)=\gamma+\mu$ is the unique solution of the last equation with the initial condition found in the first step. 

$3^{\textrm{rd}}$ step: Let $\PD_i$ be the cumulative distribution associated to $p_i$, $i=1,2$. We will prove now that if $\int_s^T\rd \PD_1\ge\int_s^T\rd \PD_2$, $s\in[0,T)$ and $\int_0^T\rd \PD_1>\int_0^T\rd \PD_2$, then $S_0[p_1](t)<S_0[p_2](t)$. For simplicity, we write $S_i=S_0[p_i]$, $i=1,2$. 
From Lemma~\ref{lem:S00}, it is clear that $S_1(0)<S_2(0)$. Furthermore, 
\[
 \left(S_1-S_2\right)'+p_1(t)(S_1-S_2)+(\mu+\alpha)(S_1-S_2)=-(p_1-p_2)S_2\le0\ .
\]
After rewriting the last equation, we find that
\[
 \frac{\rd\ }{\rd t}\left[\e^{\int_0^t(p_1+\mu+\alpha)(\tau)\rd\tau}(S_1-S_2)\right]=-\e^{\int_0^t(p_1+\mu+\alpha)(\tau)\rd\tau}(p_1-p_2)S_2\le0, 
\]
and conclude that $S_1(t)<S_2(t)$ for all $t$.

$4^{\textrm{th}}$ step: We now show that $\mathbb{E}[p_\opt]=(\mu+\alpha)(1-\langle S_0[p_\opt]\rangle)$. Indeed, let $p_n$ be a sequence such that $\int_s^T\rd\PD_n>\int_s^T\rd\PD_\opt$ for any $s\geq 0$, where $\PD_n$ and $\PD_\opt$ are the cumulative distributions associated to $p_n$ and $p_\opt$, respectively. Assume, furthermore, that $p_n\to p_\opt$ as measure.
 We use $S_n=S_0[p_n]$: therefore $S_n(t)<S_0[p_\opt](t)=\frac{\gamma+\mu}{\beta(t)}$ and then $p_n\in\strat_{\mathrm{p}}$. Furthermore, $0=\langle S_n'\rangle=\mu+\alpha-\mathbb{E}[p_n]-(\mu+\alpha)\langle S_n\rangle$ and $\mathbb{E}[p_n]=(\mu+\alpha)(1-\langle S_0[p_n]\rangle)$. We take $n\to\infty$ and use the continuity of $S_0$ in $p_n$. This finishes this step.
 
$5^{\textrm{th}}$ step: Finally, we prove that  for any $p\in\strat_{\mathrm{p}}$, $\mathbb{E}[p]>\mathbb{E}[p_\opt]$. From $S_0[p](t)<\frac{\gamma+\mu}{\beta(t)}=S_0[p_\opt](t)$, we conclude that  $\langle S_0[p]\rangle<\langle S_0[p_\opt]\rangle$ and therefore $\mathbb{E}[p]>(\mu+\alpha)(1-\langle S_0[p_\opt]\rangle)=\mathbb{E}[p_\opt]$.
\end{proof}

Now, we show that for a non-constant seasonal epidemics, it is always better to consider the natural fluctuations, also at the level of the vaccination campaign. This result goes along with~\citet{Agur_1993}. See also the discussion in~\citet{onyango2014}.
\begin{corollary}
 Let $\beta$ be a non-constant periodic function, and assume that the optimal strategy $p_\opt[\beta]$ is given by Equation~(\ref{eq:popt}). Then, $\mathbb{E}\left[p_\opt[\beta]\right]\!<\!\mathbb{E}\left[p_\opt\left[\langle\beta\rangle\right]\right]$.
 \label{cor:effortmean}
\end{corollary}

\begin{proof}
 We use the classical harmonic/arithmetic mean inequality, i.e., $\left\langle\frac{1}{\beta}\right\rangle^{-1}\le\langle\beta\rangle$, with equality if and only if $\beta$ is constant. Hence,
 \begin{align*}
  \mathbb{E}\left[p_\opt\left[\beta\right]\right]&=(\mu+\alpha)\left(1-\langle S_0\left[p_\opt\left[\beta\right]\right]\rangle\right)=(\mu+\alpha)\left(1-\left\langle\frac{\gamma+\mu}{\beta}\right\rangle\right)\\
  &\le(\mu+\alpha)\left(1-\frac{\gamma+\mu}{\langle\beta\rangle}\right)=(\mu+\alpha)\left(1-S_0\left[p_\opt\left[\langle\beta\rangle\right]\right]\right)\\&=\mathbb{E}\left[p_\opt\left[\langle\beta\rangle\right]\right],
 \end{align*}
 with equality if and only if $\beta$ is constant.
\end{proof}

\begin{corollary}
 In the optimal vaccination case, the total number of vaccinations in a single season do not exceed the number of newborns plus the number of individuals that lost their immunity during the previous season, i.e, assuming the worst case scenario, where everyone got the disease or was vaccinated, i.e., $\mathbb{E}[p_\opt]<\mu+\alpha$.
\end{corollary}

\begin{proof}
 Consider $p_0$ the vaccination strategy defined in Lemma~\ref{lem:nonempty}. Then
 \[
  \mathbb{E}[p_\opt]\le\mathbb{E}[p_0]=\frac{1}{T}\int_0^Tp_0(t)S_0[p_0](t)\rd t=\frac{p_0(\mu+\alpha)}{p_0+\mu+\alpha}<\mu+\alpha\ .
 \]
\end{proof}

The preceding corollary shows that it is possible to prevent outbreaks without vaccinating the entire population, which means that it is still possible to attain the  herd immunity effect in this more restricted framework of preventive strategies.

We finish this section by stating an explicit formula for $p_\nash$ in some special cases.

\begin{theorem}
 \label{thm:nash_solution}
Assume $r<\inf_{t\in[0,T]}\beta(t)$,
\begin{align}
\label{thm:nash_solution:eq1}
&\frac{\beta'(t)}{\beta(t)}\le\gamma+\mu\le\frac{\beta'(t)}{\beta(t)}+\beta(t)\qquad\text{and}\\
\label{thm:nash_solution:eq2}
&\frac{\rd\ }{\rd t}\left[\e^{(r+\alpha-\gamma)t}\frac{\rd\ }{\rd t}\left(\frac{\e^{(\gamma+\mu)t}}{\beta(t)}\right)\right]\le \e^{(r+\mu+\alpha)t}\left(\mu+\alpha-\frac{\alpha r}{\beta(t)}\right)\ .  
\end{align}
Then, the strategy given by

\begin{equation}
\label{eq:p_nash}
\begin{array}{ll}
& p_\nash(t)\\
&\quad=\displaystyle\frac{\beta^2(t)}{(\gamma+\mu)\beta(t)\!-\!\beta'(t)}\!\left[\mu\!+\!\alpha\!+\!(\gamma+\mu)\frac{\beta'(t)}{\beta^2(t)}-2\frac{(\beta'(t))^2}{\beta^3(t)}+\frac{\beta''(t)}{\beta^2(t)}-\frac{\alpha r}{\beta(t)}\right]\\
&\qquad-[r+\mu+\alpha]\ 
\end{array}
\end{equation}
is a Nash-equilibrium strategy.
\end{theorem}

\begin{proof}
We show that $p_\nash\ge 0$, $S_1(t),I_1(t)\in[0,1]$, $\forall t$ and $\beta(t)I_1[p_\nash](t)=r$ for all $t\ge 0$. 
Initially, let us define 
\[
S_1(t)=\frac{\gamma+\mu}{\beta(t)}-\frac{\beta'(t)}{\beta^2(t)}=\e^{-(\gamma+\mu)t}\frac{\rd\ }{\rd t}\left(\frac{\e^{(\gamma+\mu)t}}{\beta(t)}\right)\ .
\]
Note that Equation~(\ref{thm:nash_solution:eq1}) guarantees that $S_1(t)\in[0,1]$ for all $t$. Also, the assumption on $r$ implies that $I(t)\in[0,1]$ for all $t$. 
Furthermore, 
\[
 p_\nash(t)=\frac{\mu+\alpha-S_1'(t)-\alpha I_1(t)}{S_1(t)}-(r+\mu+\alpha)\ .
\]
With this definition, note that
\begin{equation*}
 -\beta(t)S_1I_1=S_1'+p_\nash S_1-(\mu+\alpha)(1-S_1)+\alpha I_1=-rS_1
\end{equation*}
and therefore $\beta(t)I_1[p_\nash](t)=r>0$. Finally, we use Equation~(\ref{thm:nash_solution:eq2}) to prove that $p_\nash\ge 0$. 
From Definition~\ref{def:nash_equilibrium}, we conclude that $\rho[p_*,p_\nash]$ does not depend on $p_*$ and therefore $p_\nash$ is a Nash-equilibrium.
\end{proof}

Note that the definition of $S_1$ was motivated by plugging $I(t)=r/\beta(t)$ into Equation~(\ref{systemI});
in particular, this choice of $I(t)$ cancels the first term of $\rho$ on Definition~\ref{def:nash_equilibrium}.

\subsection{Comparisons }

Now, we compare the vaccination effort associated with the two extreme vaccination strategies. 

\begin{proposition}\label{lem:desigual}
Assume that $p_\nash\not\in\strat_{\mathrm{p}}$. Then, $\mathbb{E}[p_\nash]<\mathbb{E}[p_\opt]$. 
\end{proposition}

\begin{proof}
We use that $0=\langle\beta I_1 S_1\rangle-(\gamma+\mu)\langle I_1\rangle$ and $0=(\mu+\alpha)(1-\langle S_1\rangle)-\mathbb{E}[p_\nash]-\alpha\langle I_1\rangle-\langle\beta I_1S_1\rangle$ to conclude that $\mathbb{E}[p_\nash]=(\mu+\alpha)(1-\langle S_1\rangle)-(\alpha+\mu+\gamma)\langle I_1\rangle< (\mu+\alpha)(1-\langle S_1\rangle)$. Assume $I_1(t)\ne 0$ for some $t$, and consequently $I_1(t)\ne 0$ for all $t$. Therefore
\[
 S_1[p_\nash](t)=\frac{\gamma+\mu}{\beta(t)}+\frac{I_1'[p_\nash](t)}{\beta(t) I_1[p_\nash](t)}=\frac{\gamma+\mu}{\beta(t)}+\frac{I'_1[p_\nash](t)}{r}\ ,
\]
and finally $\langle S_1[p_\nash]\rangle= (\gamma+\mu)\langle\beta^{-1}\rangle=\langle S_0[p_\opt]\rangle$ and $\mathbb{E}[p_\opt]=(\mu+\alpha)(1-\langle S_0\rangle)$.
\end{proof}

\begin{corollary}\label{cor:nash_opt_comp}
 Let $p_\nash$ and $p_\opt$ be the Nash and optimal vaccination strategies associated to a given transmission parameter $\beta$, respectively. Assume further that conditions in Theorem~\ref{thm:opt_solution} and~\ref{thm:nash_solution} are satisfied. Then $\mathbb{E}[p_\opt]-\mathbb{E}[p_\nash]=r(\gamma+\mu+\alpha)\langle\beta^{-1}\rangle>0$.
\end{corollary} 

\begin{proof}
 From Equation~(\ref{system}), we have $\mathbb{E}[p_\opt]=\langle p_\opt S_0\rangle =(\mu+\alpha)(1-\langle S_0\rangle)$, where $\langle S_0\rangle=(\gamma+\mu)\langle\beta^{-1}\rangle$. Furthermore, 
 \[
 \mathbb{E}[p_\nash]=\langle p_\nash S_1\rangle=
 (\mu+\alpha)(1-\langle S_1\rangle)-\langle\beta I_1S_1\rangle-\alpha\langle I_1\rangle\ . 
 \]
 We also have that 
 \[
 \left\langle S_1\right\rangle =\left\langle\left(\frac{I_1'}{\beta I_1}+\frac{\gamma+\mu}{\beta}\right)\right\rangle=\frac{1}{r}\langle I_1'\rangle+(\gamma+\mu)\langle\beta^{-1}\rangle=(\gamma+\mu)\langle\beta^{-1}\rangle.
\]
On the other hand, $\langle\beta S_1I_1\rangle=r(\gamma+\mu)\langle\beta^{-1}\rangle$. Furthermore, 
$\langle I_1\rangle=r\langle\beta^{-1}\rangle$.
Finally,
\begin{align*}
 \mathbb{E}[p_\nash]&=(\mu+\alpha)\left(1-(\gamma+\mu)\langle\beta^{-1}\rangle\right)-r(\gamma+\mu)\langle\beta^{-1}\rangle-\alpha r\langle\beta^{-1}\rangle\\
 &=\mathbb{E}[p_\opt]-r(\gamma+\mu+\alpha)\langle\beta^{-1}\rangle\ .
\end{align*}
 \end{proof}

In this work, instead of minimizing the vaccination effort in the set of strategies with $\mathcal{R}_0$ below one, we minimized the vaccination effort in a subset  $\chi_p$, which we have called the set of preventive strategies. Still, we can prove the following result.   

\begin{proposition}
 $\mathcal{R}_0[p_\opt]\le 1$. 
\end{proposition}

\begin{proof}
 The inequality follows immediately from the Definitions~\ref{deff:preventive} and \ref{def:popt} and Equation~(\ref{def:R0}). 
\end{proof}

\section{Examples}\label{sec:examples}

\subsection{The constant case}

Let us consider $\beta(t)=\beta_0$, for all $t$.

If $\beta_0\le\gamma+\mu$, we have  that $S_0(t)=1$ for all $t$, and therefore $\{p_0\bydef p(t)=0,\forall t\}\in\strat_{\mathrm{p}}$. As $\mathbb{E}[0]=0$, we conclude that
$p_\opt(t)=0$. 
As there is no endemic solution, we conclude that $p_\nash(t)=0$.

Now, assume $\beta_0>\gamma+\mu$ and assume additionally that $r>0$ is small. 
From Theorem~\ref{thm:opt_solution},  $p_\opt=(\mu+\alpha)(S_0[p_\opt]^{-1}-1)=(\mu+\alpha)\left(\frac{\beta_0}{\gamma+\mu}-1\right)$ which coincides with the optimal strategy in the traditional sense of $\mathcal{R}_0[p_\opt]= 1$.
The first condition on Theorem~\ref{thm:nash_solution} is trivially satisfied and the second one is satisfied whenever
\[
r\le r_*\bydef \frac{(\mu+\alpha)(\gamma+\mu)}{\gamma+\mu+\alpha}\left(\frac{\beta_0}{\gamma+\mu}-1\right)\ .
\]
In this case, $I_1[p_\nash]=\frac{r}{\beta_0}\ne 0$, and therefore $0=I_1[p_\nash]'=$ \linebreak $I_1[p_\nash](\beta_0S_1[p_\nash]-\gamma-\mu)$ and therefore $S_1[p_\nash]=\frac{\gamma+\mu}{\beta_0}$. Furthermore 
$0=S_1[p_\nash]'(t)=\mu+\alpha-\alpha I_1[p_\nash]-\beta_0I_1[p_\nash]S_1[p_\nash]-p_\nash S_1[p_\nash]-(\mu+\alpha) S_1[p_\nash]$. Finally, 
\[
p_\nash=(\mu+\alpha)\left(\frac{\beta_0}{\gamma+\mu}-1\right)-r\left(1+\frac{\alpha}{\gamma+\mu}\right)=p_\opt-r\left(1+\frac{\alpha}{\gamma+\mu}\right)<p_\opt\ ,
\]
i.e., the rational level of vaccination will not be able to eliminate the disease as previously shown in \citep{bauch2004}. It is clear that both $p_\opt$ and $p_\nash$ are admissible.

\subsection{The sinusoidal case}
	For the sinusoidal case, with $\beta(t)=\beta_{0}(1+\eps\cos t)$ we can provide precise results. First, note that the condition \eqref{eq:optcond} in Theorem~\ref{thm:opt_solution} is 
	\[
	-\beta_0\eps\sin t\ge -(\mu+\alpha)\beta_0(1+\eps\cos t)\left(\frac{\beta_0(1+\eps\cos t)}{\gamma+\mu}-1\right)\ ,
	\]
	or, equivalently 
	\[
	\eps\sin t\le\frac{(\mu+\alpha)\beta_0}{\gamma+\mu}(1+\eps\cos t)^2-(\mu+\alpha)(1+\eps\cos t)\ .  
	\]
	If
	\[ 
	 \eps\le(\mu+\alpha)\left(\frac{\beta_0}{\gamma+\mu}-1\right)-\eps(\mu+\alpha)\left(\frac{2\beta_0}{\gamma+\mu}-1\right)\ ,
	\]
	condition \eqref{eq:optcond} will be satisfied for every $t$. It follows that last equation is true if
	\[
	\eps\le\eps_0\bydef\left(\frac{2(\mu+\alpha)\beta_0}{\gamma+\mu}+1-\mu-\alpha\right)^{-1}(\mu+\alpha)\left(\frac{\beta_0}{\gamma+\mu}-1\right).
	\]

	Now, assume $\beta(t)=\beta_{0}(1+\eps\cos t)$ with $\eps\le\eps_0$. Note that
	\[
	 \langle S_0[p_\opt]\rangle=\frac{\gamma+\mu}{2\pi}\int_0^{2\pi}\frac{\rd t}{\beta(t)}=\frac{\gamma+\mu}{2\pi\beta_0}\int_0^{2\pi}\frac{\rd t}{1+\eps\cos t}=\frac{\gamma+\mu}{\beta_0\sqrt{1-\eps^2}}\ ,
	\]
	showing that the vaccination effort for the optimal solution 
	\[
	\mathbb{E}[p_\opt]=(\mu+\alpha)(1-\langle S_0\rangle)
	\] 
	decreases with the oscillation amplitude. Furthermore, 
	\begin{align*}
	&p_\opt(t)\\
	&\quad=(\mu+\alpha)\left(\frac{\beta_0(1+\eps\cos t)}{\gamma+\mu}-1\right)-\frac{\eps\beta_0\sin t}{\beta_0(1+\eps\cos t)}\\
	&\quad=(\mu+\alpha)\left(\frac{\beta_0}{\gamma+\mu}-1\right)+\eps\left(\frac{(\mu+\alpha)\beta_0}{\gamma+\mu}\cos t-\sin  t\right)+\mathcal{O}\left(\eps^2\right)\\
	&\quad=(\mu+\alpha)\left(\frac{\beta_0}{\gamma+\mu}-1\right)+\eps\frac{\sqrt{(\mu+\alpha)^2\beta_0^2+(\mu+\gamma)^2}}{(\mu+\gamma)}\cos(t+\varphi)+\mathcal{O}\left(\eps^2\right)\ ,
	\end{align*}
where $\varphi=\mathop{\mathrm{arctg}}\frac{\gamma+\mu}{\beta_0(\mu+\alpha)}$.
	It is clear that $\langle\frac{\beta'}{\beta}\rangle=0$ and therefore $\langle p_\opt\rangle$ is the same as in the constant case. However, the vaccination effort $\langle p_\opt S_0[p_\opt]\rangle=(\mu+\alpha)\left(1-\langle S[p_\opt]\rangle\right)=(\mu+\alpha)\left(1-\frac{\gamma+\mu}{\beta_0\sqrt{1-\eps^2}}\right)$ is strictly smaller than in the constant case. Furthermore, to first order in $\eps$, the
	optimal vaccination strategy $p_\opt$ lags behind the transmission rate by a phase shift of $\varphi$. In particular, if the birth/mortality rate is very high (and therefore there is an extremely  fast renewal of susceptible individuals), i.e., $\mu\to\infty$, then $\varphi\to\mathop{\mathrm{arctg}}\frac{1}{\beta_0}$. This means that the optimal vaccination time  shift will depend on the average transmission rate. In the more realistic case $\mu\to0$, i.e, when the renewal is low,  then $\varphi\to\mathop{\mathrm{arctg}}\frac{\gamma}{\beta_0\alpha}$. This case is of much higher interest in practice as it models the case when the disease time scales (contact rate, recovery time, temporary immunity) are much smaller than the typical time scale of one generation of the population. Assuming $\frac{\gamma}{\beta_0\alpha}\approx 1$ (i.e., immunity lasts for approximately $\mathcal{R}_0\approx\alpha^{-1}$  years), then $\varphi\approx \frac{\pi}{4}$, i.e., in a seasonal epidemic,  the peak of the vaccination rate, $p$, should be approximately 1.5 months before the transmission peak. Note that this is not the peak of the instantaneous number of vaccinations, $pS$, as illustrated in Figure \ref{fig:strategies}.

Now, we use Theorem~\ref{thm:nash_solution} to obtain the Nash strategy. Condition~(\ref{thm:nash_solution:eq1}) can be rewritten
\begin{align*}
-\eps\sin t+\mathcal{O}\left(\eps^2\right)= -\frac{\eps\sin t}{1+\eps\cos t}\le\gamma+\mu&\le \beta_0+\eps(\beta_0\cos t-\sin t)+\mathcal{O}\left(\eps^2\right)\\
&=\beta_0\!+\!\eps\sqrt{\beta_0^2+1}\cos(t+\tilde\varphi)\!+\!\mathcal{O}\left(\eps^2\right),
\end{align*}
for a certain $\tilde\varphi$. This is true if, for example, $\eps\le\gamma+\mu\le\beta_0-\eps\sqrt{\beta_0^2+1}$, in particular if $\beta_0>\gamma+\mu$, and $\eps$ is small enough.

 Equation~(\ref{thm:nash_solution:eq2}) is equivalent to the non-negativeness of $p_\nash$, given by Equation~(\ref{eq:p_nash}). Let us assume that the relative risk of the vaccination is low, i.e., $r=\mathcal{O}(\eps)$. After some extensive, but straightforward calculations, we conclude that
 \begin{align*}
  &p_\nash[\beta](t)\\=&(\mu+\alpha)\left(\frac{\beta_0}{\gamma+\mu}-1\right)-r\left(\frac{\alpha}{\gamma+\mu}+1\right)\\
  &\quad+\eps\frac{\beta_0}{\gamma+\mu}\sqrt{(\gamma+\mu)^2+1}\left(\frac{\mu+\alpha}{\gamma+\mu}\cos(t+\psi)-\frac{1}{\beta_0}\sin(t+\psi)\right)\\
=  &(\mu+\alpha)\left(\frac{\beta_0}{\gamma+\mu}-1\right)-r\left(\frac{\alpha}{\gamma+\mu}+1\right)\\
  &\quad+\eps\frac{1}{(\gamma+\mu)^2}\sqrt{(\gamma+\mu)^2+1}\sqrt{\beta_0^2(\mu+\alpha)^2+(\gamma+\mu)^2}\cos(t+\varphi+\psi)\ ,
 \end{align*}
where $\tan\psi=\frac{1}{\gamma+\mu}$ and $\tan\varphi=\frac{\gamma+\mu}{\beta_0(\mu+\alpha)}$.
It is clear that if $\frac{\beta_0}{\gamma+\mu}>1$, then, for $\eps$, $r$ small enough, the Nash-equilibrium is given by $p_\nash>0$ as obtained in the previous equation. It is also clear that both $p_\opt$ and $p_\nash$ are admissible at leading order.

If $\mu\to 0$, $\beta_0\alpha/\gamma=\mathcal{O}(1)$, then $\varphi\approx\frac{\pi}{4}$; this means that $p_\nash$ will peak shortly before $p_\opt$; however amplitude oscillations are slightly larger for $p_\nash$.
	Finally, the difference between both vaccination efforts are given by
	\[
	\mathbb{E}[p_\opt]-\mathbb{E}[p_\nash]=\frac{r(\gamma+\mu+\alpha)}{2\pi\beta_0}\int_0^{2\pi}\frac{\rd t}{1+\eps\cos t}=\frac{r(\gamma+\mu+\alpha)}{\beta_0\sqrt{1-\eps^2}}\ .
  	\]
Figure~\ref{fig:strategies} shows  the peaks and oscillations of the optimal and Nash strategies, for the sinusoidal case, and  how  the
difference between  them depends on the vaccination risk.

In Figure~\ref{fig:comparison_vaccinations}, we consider three different vaccination profiles for  sinusoidal transmission: optimal vaccination, $p(t)=p_\opt[\beta]$; optimal vaccination  in the case of average transmission rate, $p=p_\opt[\langle\beta\rangle]$; and optimal vaccination in the case of maximum transmission, $p=p_\opt[\bar\beta]$  (used in Lemma~\ref{lem:nonempty}). This example also illustrates the result of Corollary~\ref{cor:effortmean}, namely, it shows that, for a given sinusoidal transmission rate $\beta$, the effort associated to $p_\opt[\langle\beta\rangle]$ is larger than the effort associated to $p_\opt[\beta]$. However, $p_\opt[\langle\beta\rangle]$ does not prevent initial outbreaks.  Additionally, we show that $p_\opt[\bar{\beta}]$ prevents the initial outbreak but with a higher associated effort.

\begin{figure}
\centering
 \includegraphics[width=0.47\textwidth]{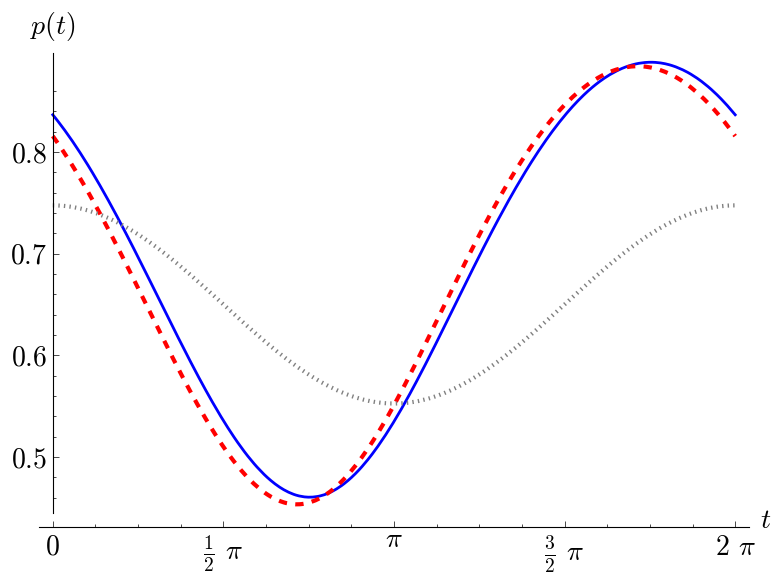}
\includegraphics[width=0.47\textwidth]{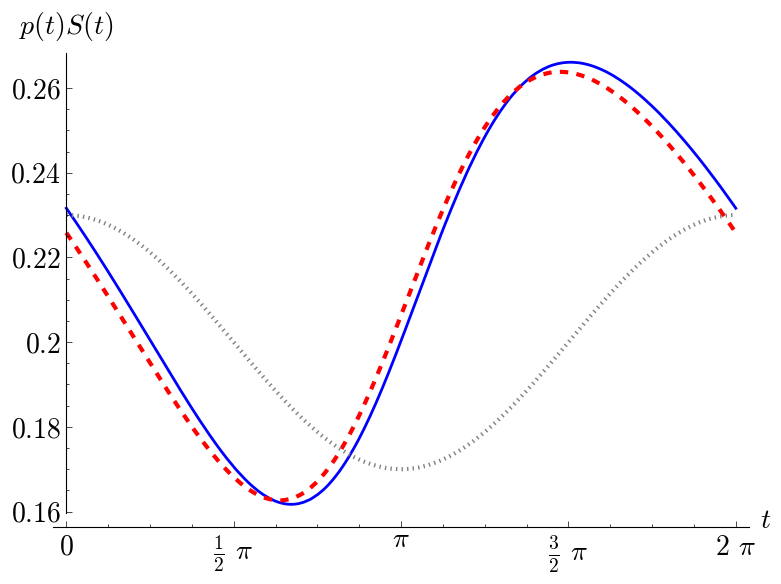}
\\
 \includegraphics[width=0.47\textwidth]{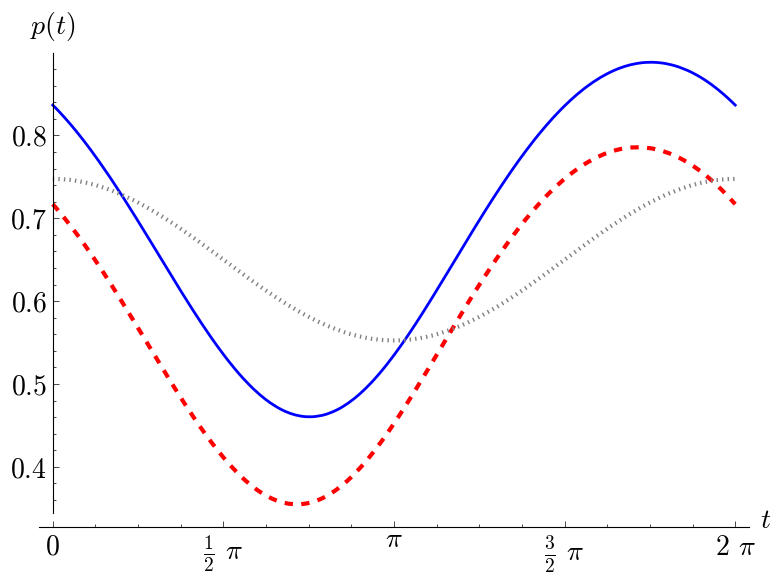}
\includegraphics[width=0.47\textwidth]{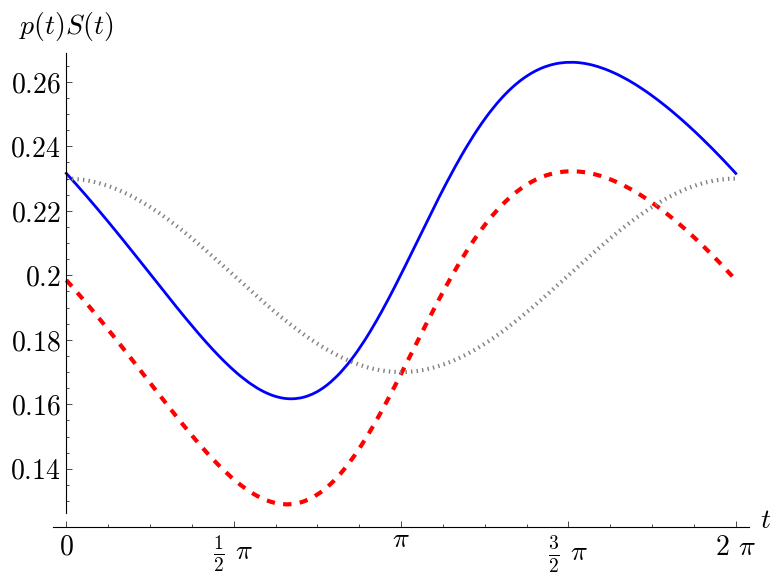}
 \caption{ Optimal and Nash strategies for the sinusoidal transmission with low and high vaccination risk, $r$. We use $\mu=(80 T)^{-1}$, $\gamma=52/T$, $\alpha=2/T$, with $T=2\pi$ and consider  $\beta(t)=26(1+0.3\cos(t))$, implying $\frac{\gamma}{\beta_0\alpha}=1$ and $\mathcal{R}_0=\frac{\langle\beta\rangle}{\gamma}=\pi\approx 3$. Left: Both $p_\opt$ (blue, continuous) and $p_\nash$ (red, dashed) oscillate in a synchronous way. The peak $p_\opt$ is $\pi/4$ before the peak of the transmission rate; the peak of $p_\nash$ is slightly before. Right: Time dependent vaccination effort, in the two cases. Above: small vaccination risk, $r=0.005$; $p_\nash$ is higher than $p_\opt$ in the beginning of the epidemic season and lower otherwise. Below: high vaccination risk, $r=0.1$; $p_\opt>p_\nash$ for all $t$, but the difference is larger when $\beta$ is decreasing. For simplicity, we plot in all cases $\beta(t)$ in gray dotted line (out of scale). The choice of parameters is inspired by influenza like epidemiological parameters: life expectancy of 80 years, recovery time of 1 week, and temporary immunity of 6 months, and $\mathcal{R}_0\approx 3$ (cf.~\citep{Goeyvaerts2015}). In the absence of vaccination, there is only one stable attractor of the solution $(S(t),I(t))$ with period $T$ (cf.~\citep{Kuznetsov_Piccardi}).
 }
\label{fig:strategies}
\end{figure}

\begin{figure}
 \centering
 \includegraphics[width=0.8\textwidth]{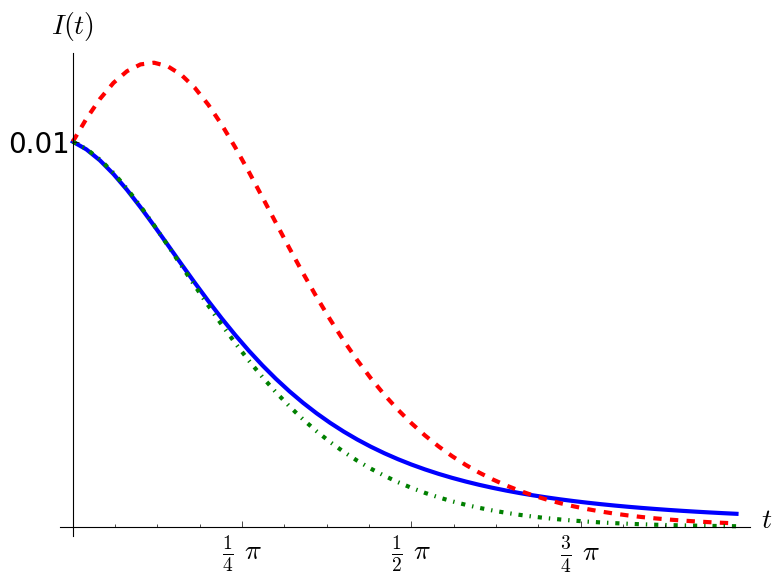}
 \caption{Prevention of initial outbreaks for different vaccination profiles.
We consider all parameters, including the transmission rate, as in Figure~\ref{fig:strategies} and assume initial conditions given by $S_0[p](0)-0.01$ and $I(0)=0.01$, where $S_0[p](0)$ is given by Lemma~\ref{lem:S00}. We consider three different vaccination profiles: $p=p_\opt[\beta]$ (blue, continuous), $p=p_\opt[\langle\beta\rangle]$ (red, dashed) and $p=p_\opt[\bar\beta]$ (green, dash-dot). Note that the optimal strategy in the case of the average transmission rate is unable to prevent initial outbreaks (however, $I(t)\to 0$ when $t\to\infty$). Vaccination efforts $\mathbb{E}[p_\opt[\bar\beta]]=0.1165> \mathbb{E}[p_\opt[\langle\beta\rangle]=0.1098>\mathbb{E}[p_\opt[\beta]]=0.1090$, respectively.
 }
 \label{fig:comparison_vaccinations}
\end{figure}

\subsection{A critical case}

Condition \eqref{eq:optcond} in Theorem~\ref{thm:opt_solution}, provides a lower bound on the derivative of the transmission coefficient $\beta$. In particular, if $\beta$ is increasing, Theorem~\ref{thm:opt_solution} provides (at least, in principle) one optimal strategy. However, to apply Theorem~\ref{thm:opt_solution}, it is important that $\beta$ does not decrease instantaneously. In this section, we will study a critical example for optimal vaccinations, where $\beta$ satisfies the critical condition  $\beta'=-(\mu+\alpha)\beta\left(\frac{\beta}{\gamma+\mu}-1\right)$ in $(0,2\pi)$ and is not differentiable at $t=0$. As usual, $\beta$ is periodic in $\R$. Explicitly, we consider
\begin{equation}\label{eq:critical}
\beta(t)= \frac{\gamma+\mu}{1-K\exp(-(\mu+\alpha)t))},\qquad t\in[0,2\pi)\ ,
\end{equation}
where $K\in(0,1)$ is a constant.

We consider a mollification of $\beta$, $\beta_\eps$, such that $\beta_\eps$ is differentiable, satisfies the condition \eqref{eq:optcond}, and $\beta_\eps\to\beta$ when $\eps\to 0$ pointwise. Let $p_\eps\bydef p_\opt[\beta_\eps]$. It is clear that $p_\eps(t)\to 0$ for $t\ne 0$, and
\begin{align*}
 \int_0^{2\pi} p_\eps\rd t&\to (\mu+\alpha)2\pi\left(\frac{1}{\gamma+\mu}\langle\beta\rangle-1\right)\\&=2\pi(\mu+\alpha)\left(\frac{1}{2\pi(\mu+\alpha)}\log\frac{\e^{2(\mu+\alpha)\pi}\!-\!K}{1-K}-1\right)\\
 &=\log\frac{\e^{2(\mu+\alpha)\pi}\!-\!K}{1-K}-2\pi(\mu+\alpha)=\log\frac{\e^{2(\mu+\alpha)\pi}-K}{1-K}+\log\e^{-2\pi(\mu+\alpha)}\\
 &=\log\frac{1-K\e^{-2(\mu+\alpha)\pi}}{1-K}=:\Gamma .
\end{align*}
Thus $p_\eps\to \Gamma\sum_{i\in\mathbb{Z}}\delta_{2\pi i}:=p_\opt$. From Equation~(\ref{eq:popt}), it is clear that for $t\in(0,2\pi)$, $p_\opt[\lim\beta_\eps](t)=\lim p_\opt[\beta_\eps](t)$, and therefore $p_\opt[\lim\beta_\eps]=\lim p_\opt[\beta_\eps]$ as measures. This shows that discontinuities in $\beta$ will be associated to peak vaccinations. 
From Lemma~\ref{lem:S00}, we have that
\[
 S[p_\opt](0)=\frac{(\mu+\alpha)\int_0^T\e^{-(\mu+\alpha)(T-s)-\frac{\Gamma}{2}-\frac{\Gamma}{2}\theta(-s)}\rd s}{1-\e^{-\Gamma-(\mu+\alpha)T}}=\frac{\e^{-\Gamma/2}(1-\e^{-(\mu+\alpha)T})}{1-\e^{-\Gamma-(\mu+\alpha)T}}\ ,
\]
where $\theta(s)=0$ for $s\le0$ and $\theta(s)=1$ for $s>0$ is the Heaviside function. 
Furthermore
\[
 \int S[p_\eps](t)p_\eps(t)\rd t\to \Gamma\e^{-\Gamma/2}\frac{1-\e^{-(\mu+\alpha)T}}{1-\e^{-\Gamma-(\mu+\alpha)T}}\le (\mu+\alpha)T\ ,
\]
where we used that 
\[
x\e^{-x/2}\frac{1-\e^{-y}}{1-\e^{-x-y}}=x\frac{\mathop{\mathrm{sinh}}\frac{y}{2}}{\mathop{\mathrm{sinh}}\frac{x+y}{2}}\le y\ , \forall x,y>0\ .
\]
This last inequality follows from the convexity of the function 
\[
\mathop{\mathrm{sinh}}\alpha=\mathop{\mathrm{sinh}}\left(\frac{\beta}{\alpha+\beta}0+\frac{\alpha}{\alpha+\beta}(\alpha+\beta)\right)\le\frac{\beta}{\alpha+\beta}\mathop{\mathrm{sinh}}0+\frac{\alpha}{\alpha+\beta}\mathop{\mathrm{sinh}}\left(\alpha+\beta\right)\ ,
\]
and therefore, for $\alpha,\beta> 0$,
\[
 \frac{\mathop{\mathrm{sinh}}\alpha}{\mathop{\mathrm{sinh}}(\alpha+\beta)}\le\frac{\alpha}{\alpha+\beta}\le\frac{\alpha}{\beta}\ .
\]
Optimal strategy for the critical case, using a relaxed version of the transmission function $\beta(t)$,  is illustrated in  Figure~\ref{fig:critical_vaccinations}.

\begin{figure}
 \centering
 \includegraphics[width=0.47\textwidth]{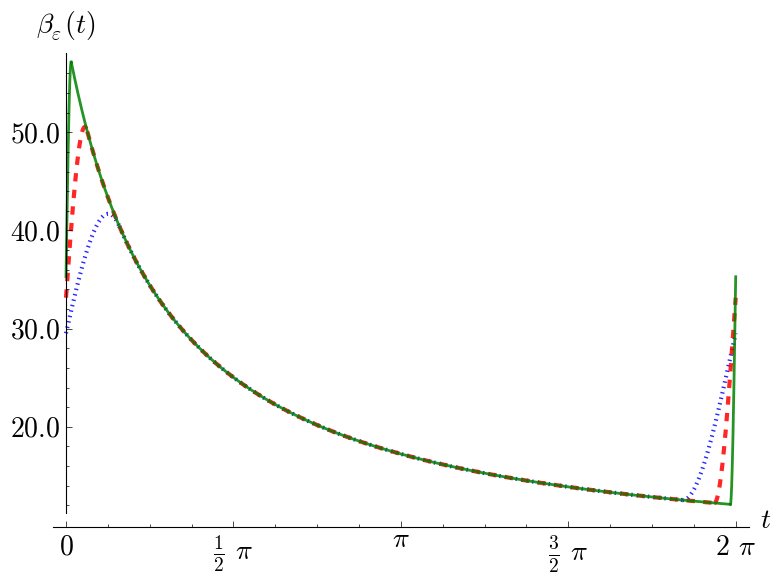}
 \includegraphics[width=0.47\textwidth]{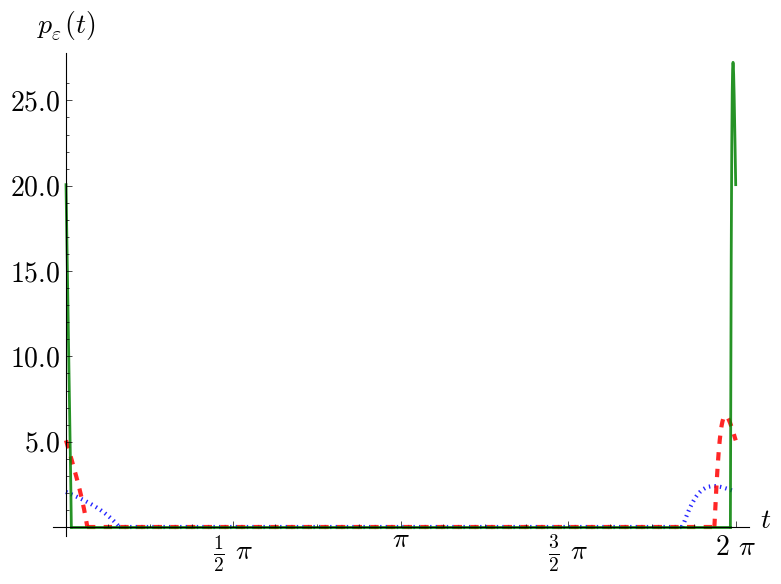}
 \caption{Optimal strategy for the critical case, with relaxed $\beta(t)$ Left: Relaxed version of $\beta(t)$ given by Equation~\eqref{eq:critical} with $K=0.862$, and the other parameters as in Figure~\ref{fig:SIR}, with $T=2\pi$. We consider a sequence $\beta_\eps$ of differentiable functions, such that $\beta_\eps(t)=\beta(t)$ for $t\in(\eps,2\pi-\eps)$ and the points $2\pi-\eps$ and $2\pi+\eps\equiv\eps$ are connected by a third order polynomial. Right: a sequence of $p_\eps\bydef p_\opt[\beta_\eps]$; note that $\left(p_\eps\right)_{\eps>0}$ resembles a delta-sequence.}
 \label{fig:critical_vaccinations}
\end{figure}

\section{Discussion}

In this work, we consider a SIRS model with periodic transmission, where we introduce periodic vaccination of adults. We are naturally led to consider temporary immunity, as discussed in~\citep{onyango2014}: in the case of diseases with life-long immunity and lifespan much longer than the period ($\mu\sim \alpha \ll {1}/{T}$) the vaccination is not affected by periodicity. For instance, measles has periodic transmission rate but, since it confers permanent immunity,  periodic vaccination is not used.
In this work we implicitly assume $\mu \ll  {1}/{T} \approx \alpha$.

We study the consequences of two extreme types of vaccination strategies: mandatory vaccination, where a certain predefined fraction of the population is vaccinated; and voluntary vaccination, where individuals can choose freely to be vaccinated or not, according to their risk perception. 
Classically, the objective is to minimize the vaccination effort while reducing the effective reproductive number below one, which guaranties long term disease elimination. Here, we choose to work with an alternative definition of optimal vaccination. We define a class of preventive vaccination strategies as vaccination profiles that, for any sufficiently small perturbation of the disease free state, the number of infectious individuals is monotonically decreasing, avoiding the occurrence of any epidemic event. This approach allows, for specific regular transmission functions $\beta(t)$, the derivation of an analytical expression of the optimal strategy. In general, we prove the existence of an optimal strategy, in a suitably defined closure of the space of all preventive strategies, which minimizes the vaccination effort.

In this work, we extend the classical results by \citet{bauch2004} to periodic functions, based on a series of recent results on periodic diseases.  We model human behaviour using classical economical theory, where individuals are assumed to be rational and fully informed.
We define the set of vaccination strategies that provide herd immunity, for which the rational strategy of a given focal individual is not to be vaccinated. Finally, we prove the existence of a Nash vaccination strategy as the strategy that minimizes the joint risk for every individual, taking into account the strategy of all other individuals.

In general, both optimal and Nash strategies will not be functions but Radon measures. For specific forms of the transmission rate, we provide explicit formulas, which includes some important examples as constant or sinusoidal transmission functions.

There are several natural limitations of the work presented here. One first limitation is that we consider only the stationary solution of the System~(\ref{system}), but we never discuss the approach to this equilibrium. This is an important question, both in the study of ordinary differential equations (i.e., the study of the basin of attraction) and in evolutionary game theory, where the study of $\omega$-limits of conveniently defined dynamical equations is preferred to the static study of Nash-equilibria. In the non-stationary case, a rational decision will require the ability to forecast the evolution of the epidemic, i.e., rational decisions will depend on future decisions of the entire population and not only on the past decisions. This is mathematically described by the so called ``mean field game theory''~\citep{Lasry_Lions_2007} and will be object of a future work.

Closely related ideas will also help us to solve one of the major gaps of the current work: the lack of a numerical method for finding Nash-equilibria solution when Theorem~\ref{thm:nash_solution} fails. More precisely, the idea will be to develop a numerical method that allows constant update in individual decisions and, consequently, also at the population level. As discussed before, this will require, at the individual level, a certain expectation on the future evolution of the disease. 

It is also important, and will be subject of a future work, to design a precise scheme, possibly numerical, that allows to go beyond Theorem~\ref{thm:opt_solution}. This will require the use of Optimal Control Theory. In fact, given $p$, it is possible to explicitly obtain the disease free solution $S[p](t)$ (see Lemma~\ref{lem:S00}) and therefore we need to minimize $\int_0^Tp(t)S[p](t)\rd t$ in $\strat_{\mathrm{p}}[\beta]$. (Equivalently, we may maximize $\int_0^TS[p](t)\rd t$ in the same set, as $\langle pS_0\rangle=(\mu+\alpha)\left(1-\langle S_0\rangle\right)$.) We also plan to compare $p_\opt$ with the optimal solution in the approach in which the vaccination effort is minimized in the class of vaccination functions $p$ such that $\mathcal{R}_0\le 1$.

Furthermore, despite the simplicity of the periodic SIRS system (even with vaccination), solutions can be extremely complicated; even chaotic solutions may be present in such simple systems, cf.~\citep{Kuznetsov_Piccardi}. The coupling of the differential equations with human rational behaviour presented in this work only started the exploration of all this mathematical richness.

\appendix

\section{Proof of Lemma~\ref{lem:sol_per}}\label{ap:proof1}

We follow closely the proof at~\citep{rebelo2012}, where $x_1=I$, $x_2=S$ and $x_3=R$. Also, $m=1$ indicates that there is only one infectious class and $n=3$ denotes the three possible classes in the model. We readily verify that conditions $(A_1)-(A_5)$ in~\citep{rebelo2012} 
are satisfied. 
Uniqueness and stability (in the disease free subspace) of the disease free solution $(S_0(t),0,R_0(t))$ is guaranteed by standard theorems.
The linearisation of system~\eqref{system} restricted to $I=0$ around the disease-free solution $(S_0(t),0,R_0(t))$ is given by
\[
	\left\{
		\begin{array}{l}
			s'=-\mu s -p(t) s + \alpha r \\
			r'=-\mu r +p(t) s - \alpha r
		\end{array}
	\right.
\]
that can be explicitly solved to get
\[
	\left\{
		\begin{array}{l}
			s(t)=s(0)a(t)+(s(0)+r(0))b(t)\\
			r(t)=(s(0)+r(0))e^{-\mu t}-s(t)
		\end{array}
	\right.
\]	
where $a(t)=e^{-\int_0^t \mu+\alpha+p(\tau)\,d\tau}$ and $b(t)=\int_0^t \alpha e^{-\mu\tau -\int_\tau^t(\mu+\alpha+p(l))\,\rd l}\,\rd\tau$. This yields that the monodromy matrix $M(t)$ of the linearised system is
\[
M(T)=
	\left[
		\begin{array}{cc}
			a(T)+b(T)& b(T)\\
			e^{-\mu T}-a(T)-b(T)& e^{-\mu T}-b(T)
		\end{array}
	\right]\ .
\]	
We compute the Floquet multipliers $\rho_1=a(T)<1$ and $\rho_2=e^{-\mu T}<1$ and conclude condition $(A_6)$.
We verify immediately that conditions $(A_7)$ and $(A_8)$ are also satisfied.

Let $(S,I,R)$ be a solution of the system~\eqref{system} and $(S_0,0,R_0)$ the disease free solution of the same system. Therefore
\begin{equation}\label{Sbound1}
(S-S_0)'=-\alpha I-\beta IS -(p(t)+\mu +\alpha)(S-S_0)\leq  -(p(t)+\mu+\alpha)(S-S_0).
\end{equation}
This yields for every $t\geq 0$. Consequently, by Gronwall's lemma 
\[
S(t)-S_0(t)\leq (S(t_0)-S_0(t_0)) \exp \left(-\int_{t_0}^t(p(s)+\mu+\alpha)\,\rd s\right)
\]
for any $t\geq t_0 \geq 0$.

So, for any $\epsilon>0$ there is $t_1(\epsilon)\geq 0$ such that for any $t>t_1$ we have
\begin{equation}\label{x2-x2*}
S(t)-S_0(t) < \epsilon.
\end{equation}

Now, assume that there is $t_0\geq 0$ such that $I(t)\leq \epsilon$ for every $t\geq t_0$. Therefore, as $S<1$,
\begin{equation}\label{Sbound2}
(S_0-S)'=\alpha I+\beta IS -(p(t)+\mu + \alpha)(S_0-S) \leq  (\beta+\alpha)\epsilon-(p(t)+\mu +\alpha)(S_0-S) ,
\end{equation}
and, by~\citet[lemma~1]{rebelo2012}
there exists $k>0$, independent of $\epsilon$, and $t_2(\epsilon)\geq t_0$ such that for all $t\geq t_2$
\begin{equation}\label{x2*-x2}
S_0(t)-S(t) \leq k \epsilon.
\end{equation}

For any solution in the disease-free subspace (i.e., with $I(t)=0$ for all $t\ge 0$), we have the validity of 
 Conditions~(\ref{Sbound1}) and~(\ref{Sbound2}), and, therefore, we conclude that the disease free solution is globally asymptotically stable.

\smallskip

Now, we show the alternative in Lemma~\ref{lem:sol_per}. Let $\mathcal{R}_0=\frac{\langle \beta S_0\rangle}{\gamma+\mu}$, as defined in~\citet{onyango2014} and in a more general setting in~\citet{rebelo2012,wang2008}. We show that the only relevant assumption in~\citet[theorem 2]{rebelo2012} is the value of $\mathcal{R}_0$. In particular, we define the $1\times 1$ matrices
$F(t)=[\beta(t)S_0(t)]$ and $V(t)=[\mu + \gamma]$~\citep{rebelo2012}. By~\eqref{x2-x2*}, for every $\epsilon>0$ there is $t_1>0$ such that for $t \geq t_1$,
\begin{align*}
I'(t)&=	\beta(t)I(t)S(t)-(\gamma+\mu)I(t)
			\leq \beta(t)I(t)(\epsilon+S_0(t))-(\gamma+\mu)I(t)\\ 
			&\leq \left(\frac{F(t)}{\lambda_1(\epsilon)}-V(t)\right)I(t)
\end{align*}
where  $\lambda_1(\epsilon)=\min_t \frac{S_0(t)}{\epsilon+S_0(t)}$. Notice that $\lambda_1(\epsilon)>0$ and that $\lambda_1(\epsilon) \to 1$ from below when $\epsilon \to 0$.

If $\mathcal{R}_0<1$, \citet[theorem~2, condition 1]{rebelo2012} guarantees that the diseases dies out, $I(t) \to 0$ as $t \to \infty$ and that $(S_0,0,R_0)$ is globally asymptotically stable.

Now, assume $\mathcal{R}_0>1$; from $\frac{1}{T}\int_0^T\beta(t)S_0(t)\,\rd t-(\mu + \gamma)>0$ and by continuity in $t$,  we have that $F(t)-V(t)$ is irreducible for some $t \in [0,T)$.

For any $\epsilon>0$, if there is $t_0 \geq 0$ such that $I(t) \leq \epsilon $ for every $t\geq t_0$ then, by~\eqref{x2*-x2}, there is $t_2\geq t_0$ such that, for $t \geq t_2$,
\[
I'(t)\geq\beta(t)I(t)(S_0(t)-k\epsilon)-(\mu+\gamma)I(t) \geq\left(\frac{F(t)}{\lambda_2(\epsilon)}-V(t)\right)I(t),
\]
where $\lambda_2(\epsilon)=\max_t\frac{S_0(t)}{S_0(t)-k\epsilon}$ satisfies $\lim_{\epsilon \to 0^+}\lambda_2(\epsilon)=1$ and $\lambda_2:(0,\epsilon^*) \to \mathbb{R}^+$ if we choose $\epsilon^*\leq \frac{1}{k}\min_t S_0(t)$ (observe that Lemma~\ref{lem:S00} and its proof guarantees that $\min_t S_0(t)>0$). 

We conclude that the conditions  in \citet[theorem~2, statement~2]{rebelo2012} are satisfied and there is uniform persistence of system~\eqref{system} with respect to $I$.

\smallskip

Finally, we prove the existence of a persistent periodic solution.
We define the $T$-mapping $P:\Delta^{2} \to \Delta^{2}$ by $P(x_0)=x(T,(x_0,0))$, where $x(\cdot,(x_0,0))$ is the solution of System~\ref{system}  with initial conditions $x(0)=x_0\in\Delta^{2}\bydef\{(S,I)\in\mathbb{R}^2_+, S+I\le 1\}$, the two-dimensional simplex. $P$ is a continuous map such that $P(M_0) \subset M_0$ for $M_0=\{(I,S) \in \Delta^{2} : I \neq 0\}$. Observe that $M_0$ is an open set of $\Delta^{2}$ with the topology induced in $\Delta^{2}$. As we have uniform persistence of system~\eqref{system} with respect to~$I$ we also have uniform persistence of $P$ with respect to $M_0$ as described in~\citet{zhao1995} (for a more general case see~\citep{magal2005}). 
Applying \citep[theorem~2.1]{zhao1995} we conclude that $P: M_0 \to M_0$ admits a global attractor and there is a fixed point for $P$ in that attractor, which is a $T$-periodic solution of~\eqref{system} (see, for example, lemma 4.4 in~\citet{verhulst1996}).

\section{Proof of Theorem~\ref{thm:existence}}\label{ap:existence}

\begin{proof}
	First we recall that $\radon$ is compact with the weak topology, cf.~\citep{koralov2007}. We have that $C([0,T])\subset \radon$ in the sense that for each continuous function we consider the correspondent cumulative distribution function. Consequently the closure $\overline{C([0,T])}\subset\radon$ is compact. 	
	The map  $\mathbb{E}: \overline{C([0,T])} \to [0,+\infty]$ defined in Definition~\ref{def:eff} and extended in Lemma~\ref{lem:effort_measure} is continuous with respect to the weak topology. 
	
	\noindent\textbf{Existence of $p_\opt$:} As $\overline{\strat_{\mathrm{p}}}\subset\radon$ is compact, from the continuity of  $\mathbb{E}$, we conclude that there is a measure $p_\opt\in\overline{\strat_{\mathrm{p}}}$ such that for all $p\in\strat_{\mathrm{p}}$, $\mathbb{E}[p]\ge\mathbb{E}[p_\opt]$.  

	\noindent\textbf{Existence of $p_\nash$:} Let us consider fixed time intervals $\Delta t$, such that $T/(\Delta t)=N\in\N$, and consider periodic continuous  piecewise affine functions in intervals $(i\Delta t, (i+1) \Delta t)$, $i \in \{0,\ldots,N-1\}$ (i.e, functions, such that $f(t)=f(\lfloor t/\Delta t\rfloor \Delta t)+\frac{f((\lfloor t/\Delta t\rfloor+1) \Delta t)-f(\lfloor t/\Delta t\rfloor \Delta t)}{\Delta t}(t-\lfloor t/\Delta t\rfloor \Delta t)$; furthermore, $f(N\Delta t)=f(0)$). These functions can be represented by vectors in $\R_+^N$. Let $\mathbf{v}\in\Upsilon\bydef \{\mathbf{v}\in\R^{N}_+|0\le v_1\le v_2\le\dots\le v_N\le (\alpha+\mu)N\frac{\bar\beta}{\gamma+\mu}\}$. The set $\Upsilon$ is convex and compact. Now for each vector $\mathbf{v}\in\Upsilon$ consider the distributions $p_{\mathbf{v}}:[0,T]\to\R_+$, such that the associated cummulative distribution is given by $\mathsf{P}_{\mathbf{v}}(t)=v_i+\frac{v_{i+1}-v_i}{\Delta t}(t-i \Delta t)$ for $t\in[i\Delta t,(i+1)\Delta t]$, $i\in\{0,\dots,N-1\}$. Consider $I[p_{\mathbf{v}}]$ and $S[p_{\mathbf{v}}]$ solutions of System~(\ref{system}), and define for $\mathbf{v},\mathbf{v}^*\in \Upsilon$ the joint risk (except for some immaterial constants)
	\[
\rho^{\Delta t}[\mathbf{v}^*,\mathbf{v}]=-\sum_{i=1}^{N}\bigl\{\left(r-\beta(i\Delta t)I[p_{\mathbf{v}}](i\Delta t)\right)\e^{-v^*_i}\bigr\}\ .
	\]
	
	It is clear that $v_i^*\approx\int_{0}^{i\Delta t}\rd\PD_*=\PD_*([0,i\Delta t])$.
	We define a function $\mathcal{F}:\Upsilon\to 2^{\Upsilon}$ such that $\tilde{\mathbf{v}}\in\mathcal{F}[{\mathbf{v}}]$ if and only if $\rho^{\Delta t}[\tilde{\mathbf{v}},\mathbf{v}]\le\rho^{\Delta t}[\mathbf{v}',\mathbf{v}]$ for all $\mathbf{v}'\in \Upsilon$.  

	It is clear that $\mathcal{F}[\mathbf{v}]\ne\emptyset$, as $\Upsilon$ is compact and $\rho^{\Delta t}[\cdot,\mathbf{v}]$ is continuous. Now, we prove that $\mathcal{F}[\mathbf{v}]$ is closed and convex. The first property follows again from the continuity of $\rho^{\Delta t}[\cdot,\mathbf{v}]$. Define $\hat r_i=r-\beta(i\Delta t)I[p_{\mathbf{v}}](i\Delta t)$. We divide the last property in two cases:
	\begin{enumerate}
	 \item Assume that $\hat r_i\le 0$ for all $i$ and let $\mathbf{v}^*$ be such that $\rho[\mathbf{v}^*,\mathbf{v}]\le\rho[\mathbf{v}',\mathbf{v}]$ for all $\mathbf{v'}\in\Upsilon$. Assume in addition that there is $\tilde{\mathbf{v}}$ such that $\rho[\mathbf{v}^*,\mathbf{v}]=\rho[\tilde{\mathbf{v}},\mathbf{v}]$ for all $\mathbf{v}\in\Upsilon$. Therefore, for $\alpha\in(0,1)$,
	 \begin{align*}
	  \rho^{\Delta t}[\alpha\mathbf{v}^*+(1-\alpha)\tilde{\mathbf{v}},\mathbf{v}]&=-\sum_i\hat r_i\e^{-\alpha v_i^*-(1-\alpha)\tilde v_i}\\
	  &\le -\sum_i\hat r_i\left(\alpha\e^{-v_i^*}+(1-\alpha)\e^{-\tilde v_i}\right)\\
	  &=\alpha\rho^{\Delta t}[\mathbf{v}^*,\mathbf{v}]+(1-\alpha)\rho^{\Delta t}[\tilde{\mathbf{v}},\mathbf{v}]=\rho[\mathbf{v}^*,\mathbf{v}]\\
	  &\le\rho^{\Delta t}[\alpha\mathbf{v}^*+(1-\alpha)\tilde{\mathbf{v}},\mathbf{v}]\ ;
	 \end{align*}
hence $\mathbf{v}^*,\tilde{\mathbf{v}}\in\mathcal{F}[\mathbf{v}]$ implies that $\alpha\mathbf{v}^*+(1-\alpha)\tilde{\mathbf{v}}\in\mathcal{F}$.
        \item Let $\mathcal{I}\bydef\{i|\hat r_i>0\}\ne\emptyset$. In order to minimize $\rho^{\Delta t}[\cdot,\mathbf{v}]$, we impose to each $i\in\mathcal{I}$ the minimum possible value, i.e., $v^*_i=v^*_{i-1}$. Therefore, we shall minimize
        \[
         \rho_{\mathcal{I}}^{\Delta t}[\mathbf{v}^*,\mathbf{v}]\bydef -\sum_{i\not\in\mathcal{I}}\bigl\{\left(r-\beta(i\Delta t)I[p_{\mathbf{v}}](i\Delta t)\right)\e^{-v^*_i}\bigr\}\ .
        \]
The existence of a minimum is guaranteed by the compacity of $\Upsilon_{\mathcal{I}}\bydef \{\mathbf{v}\in\R^{N}_+|0\le  v_1\le v_2\le\dots\le v_N\le (\alpha+\mu)N\frac{\bar\beta}{\gamma+\mu}, i\in\mathcal{I}\Rightarrow v_{i-1}=v_i\}$. Then, we repeat the previous analysis and conclude that $\mathcal{F}[\mathbf{v}]$ is closed and convex.
	\end{enumerate}

	We conclude that the set of best replies is non-empty, convex, closed and due to the continuity of $I$ in $p$ (see~Lemma~\ref{lem:sol_per}) and of $\rho^{\Delta t}$ in $\mathbf{v}$ and $\mathbf{v}^*$,  the graph is closed.
	Therefore from standard applications of Kakutani fixed point theorem, there is a fixed point vector of the function $\mathcal{F}$, $\mathbf{v}^{(\Delta t)}$, such that its affine function continuation $p^{(\Delta t)}$ is a Nash-equilibrium restricted to affine functions with steps $\Delta t$. See, e.g., \citep{osborne1995}. 
	Furthermore
	\[
		\PD^{(\Delta t)}([0,T])=\sum_{i=0}^{N-1}p^{(\Delta t)}(i\Delta t)\Delta t\le\sum_{i=0}^{N-1}v_i\Delta t\le(\alpha+\mu)N\frac{\bar\beta}{\gamma+\mu}\Delta t=(\alpha +\mu) T\frac{\bar\beta}{\gamma+\mu}\ .
	\]
	From the compactness of $\overline{C([0,T])}$, there is a measure $p\in\overline{C([0,T])}$ such that $\lim_{\Delta t\to 0}p^{(\Delta t)}=p$ (possibly after taking subsequences), where the convergence is in the weak topology. 
	
	The last step is to prove that $p$ is indeed a Nash-equilibrium in~$\overline{C([0,T])}$. Assume it is not; then, there is $\tilde p$ in~$\overline{C([0,T])}$  such that $\rho[\tilde p,p]<\rho[p,p]$. From the continuity of $\rho$, there is a $(\Delta t)_0>0$, small enough, such that all restricted Nash-equilibria found above are such that $\rho[\tilde p,p^{(\Delta t)}]<\rho[p^{(\Delta t)},p^{(\Delta t)}]$
	for $\Delta t<(\Delta t)_0$. Let $\bar p_n$ be a sequence of continuous functions in~${C([0,T])}$ such that $p_n\to\tilde{p}$ weakly, and therefore $\rho[\bar p_n,p^{(\Delta t)}]<\rho[p^{(\Delta t)},p^{(\Delta t)}]$, for a certain value of $\Delta t$ and $n$ large enough.
	Using the fact that  $\rho^{\Delta t}[\mathbf{v}^*,\mathbf{v}]$ is the trapezoidal approximation of $\rho[p^*,p]$ (and therefore differs in $\mathcal{O}((\Delta t)^{2})$) and taking $\Delta t$ possibly even smaller, we conclude that $p^{(\Delta t)}$ is not a restricted Nash-equilibrium, contradiction.
	\color{black}
\end{proof}

\section*{Acknowledgements}

The authors are grateful to Max Souza (UFF, Brazil) and Nicolas Bacaer (IRD \& Université Paris 6, France) for stimulating discussions. We are also grateful to the two anonymous referees for their comments that allowed us to improve the manuscript.
This work was partially supported by FCT/Portugal project EXPL/MAT-CAL/0794/2013, and by Strategic Project UID/MAT/00297/2013 (Centro de Matemática e Aplicações, Universidade Nova de Lisboa).


\end{document}